\documentclass[11pt,a4paper]{article}
\usepackage[margin=1in]{geometry}
\usepackage{amsmath,amssymb,amsfonts,amsthm,mathtools}
\usepackage{graphicx}
\usepackage{booktabs}
\usepackage{float}
\usepackage{hyperref}
\usepackage{enumitem}
\usepackage{algorithm}
\usepackage{algpseudocode}
\usepackage{algorithmicx}
\usepackage{setspace}
\usepackage{natbib}
\usepackage{caption}
\usepackage{subcaption}
\usepackage{xcolor}

\newtheorem{definition}{Definition}[section]

\newtheorem{lemma}{Lemma}[section]
\newtheorem{proposition}{Proposition}[section]
\newtheorem{theorem}{Theorem}[section]
\newtheorem{corollary}{Corollary}[section]
\newtheorem{remark}{Remark}[section]

\begin{document}

\begin{center}
{\Large\textbf{Pricing Temperature-Index Insurance under Long Memory and Stochastic Time Change}}

\vspace{1em}
N. Karimi$^{a}$, \quad F. Shokrollahi$^{b,*}$

\vspace{0.5em}
{\small
\emph{\footnotesize $^{a}$ Department of Applied Mathematics, Faculty of Mathematics and Computer Science,} \\
\emph{\footnotesize Amirkabir University of Technology, No. 424, Hafez Ave., 15914, Tehran, Iran}\\
$^{b}$ Department of Mathematics and Statistics, University of Vaasa,
Vaasa, Finland\\[0.3em]
$^{*}$ Corresponding author. Email: foad.shokrollahi@uwasa.fi
}
\end{center}

\begin{abstract}
This paper develops a unit-consistent actuarial framework for pricing capped
cumulative temperature-index insurance under long-range dependence and
stochastic variability. Daily temperature anomalies are modeled as increments
of fractional Brownian motion evaluated at an operational time generated by
the integral of a stationary normalized Cox--Ingersoll--Ross process. We show
that the stochastic time change preserves stationarity and the long-memory
covariance decay of the increments while introducing additional variability
through the random operational clock. The cumulative temperature index admits
a conditionally Gaussian representation, which leads to an exact conditional
exponential kernel for capped stop-loss contracts and ensures existence of the
entropic premium for every positive risk-aversion parameter. Consequently,
valuation reduces to an outer Monte Carlo expectation over the accumulated
CIR time, avoiding fractional Brownian path simulation and covariance-matrix
construction. We further establish monotonicity properties of the premium with
respect to risk aversion and conditional volatility. An empirical illustration
based on Chicago temperature data shows that both long memory and stochastic
time change can materially affect insurance premiums relative to conventional
Brownian and fractional Brownian benchmarks, with the Hurst parameter playing
an important role in valuation uncertainty. The proposed framework therefore
provides a tractable approach for incorporating persistent dependence,
stochastic variability, and bounded insurance losses into climate-index
pricing.
\end{abstract}

\section{Introduction}\label{sec:intro}

Climate-related indices have become increasingly important instruments for
measuring, transferring, and managing exposures to temperature, precipitation,
drought, wind, and compound extremes \citep{Pan2022,Nuntasuwan2026,Yavrum2026}.
Among these instruments, temperature-index insurance is particularly attractive
when economic losses are associated with persistent departures from normal
weather conditions over a coverage period rather than with a single extreme
observation. The contractual quantity is then naturally cumulative: a sequence
of moderately adverse temperature anomalies may generate a material insured
loss even when no individual day appears exceptional. This feature makes the
dependence structure and maturity scaling of the underlying temperature process
central to actuarial valuation.

Unlike indemnity insurance, index insurance determines payment from a
pre-specified and objectively observed index rather than from an individual loss
assessment. This design can reduce claims-adjustment costs, moral hazard, and
adverse selection, while facilitating transparent and rapid settlement
\citep{Barnett2007,MirandaFarrin2012,Jensen2017,Carter2017}. Its principal
limitation is basis risk: the index realization may differ from the actual
economic loss experienced by the policyholder. Recent reviews therefore regard
index construction, contract calibration, pricing, and validation against
realized losses as closely related components of a credible climate-risk
transfer mechanism \citep{Benso2023,Zhu2025,Wijesena2025,Nuntasuwan2026}.

Recent actuarial research has also moved toward explicitly climate-conditioned
risk-transfer models. \citet{KarimiShokrollahi2026Climate} model excess-of-loss
reinsurance and catastrophe bonds through a Cox process whose catastrophe
intensity depends on a temperature-related climate state, while
\citet{KarimiSalavatiShokrollahi2026Cascade} develop a climate-conditioned
cascade framework for multi-peril reinsurance in which event propagation and
contract-layer losses depend on the evolving physical state. These studies
illustrate two important directions in climate-risk modelling: allowing climate
conditions to modify event frequency and allowing dependence across hazards to
shape insured tail losses. The present paper addresses a complementary problem.
Rather than modelling catastrophe counts or multi-peril loss cascades, we price
a contract written directly on a cumulative temperature index and focus on the
effect of persistent daily dependence and stochastic local variability on the
premium.

Basis risk remains conceptually separate from this stochastic-index valuation
problem. It may arise from spatial separation between the observation site and
the insured exposure, mismatch between the selected meteorological variable and
the loss mechanism, or misspecified coverage dates and trigger levels
\citep{Woodard2008,Dalhaus2016,Lichtenberg2022}. The contract-design literature
has addressed these issues using flexible response functions, quantile methods,
penalized splines, remote-sensing products, and machine-learning indices
\citep{Conradt2015,ConradtQuantile2015,Cesarini2021,Chen2024,Tan2024}. Such
methods improve the mapping from observed weather variables to insured losses,
but they do not eliminate the distinct actuarial requirement to specify the
probability law of the contractual index itself. For a cumulative temperature
contract, that law determines the trigger probability, expected payment, risk
loading, maturity effect, and the probability of exhausting the policy limit.

The established temperature-derivative literature typically begins by removing
deterministic seasonality and modelling the residual temperature process with
Gaussian mean-reverting dynamics \citep{Alaton2002,Cao2004,Benth2011,Benth2012}.
Subsequent contributions allow time-varying mean reversion, regime changes,
non-Gaussian innovations, stochastic volatility, or subordinated operational
time \citep{Zapranis2008,Turkvatan2020,Tong2020,Alfonsi2024,Li2026}.
Incomplete-market and data-sparse valuation methods have likewise been developed
for weather-linked claims \citep{Alexandridis2021}. These models provide a
strong foundation for seasonal temperature modelling and weather-derivative
valuation. For cumulative index insurance, however, a short-memory
specification may understate persistence when dependence among daily anomalies
decays more slowly than an exponential autocorrelation function.

Long-range dependence in atmospheric and surface-temperature records has been
studied using rescaled-range, spectral, detrended-fluctuation, and fractional
integration methods. Evidence has been reported for individual stations,
regional series, global temperature anomalies, and long climate-model
simulations \citep{KoscielnyBunde1998,Caballero2002,Rypdal2013,Ostvand2014,
Ludescher2016,Fredriksen2017,GilAlana2026}. The estimated strength of
persistence is sensitive to the data, horizon, preprocessing, and statistical
method, so long memory should be tested empirically and its parameter
uncertainty propagated rather than imposed mechanically. Its pricing relevance
is nevertheless immediate for cumulative contracts: persistent daily anomalies
change the variance--maturity relation and can materially alter lower-tail
probabilities relative to a short-memory model calibrated at a single horizon.

We represent this persistence using fractional Brownian motion (fBm)
\citep{Mandelbrot1968,Biagini2008}. A key modelling distinction is made between
fBm levels and fractional Gaussian noise. Daily deseasonalized temperature
anomalies are represented by increments of a time-changed fBm, while the
insurance index is the sum of those increments. This increment-based
construction avoids integrating an already cumulative fBm level and produces
the $T^{2H}$ variance scaling associated with persistent increments. The Hurst
parameter $H>1/2$ governs the power-law decay of dependence, whereas a separate
amplitude parameter determines the physical scale of the persistent temperature
component.

Long memory alone does not capture fluctuations in local variability. We
therefore evaluate fBm at the stochastic operational time
\[
    \tau_t=\int_0^t v_s\,\mathrm{d}s,
\]
where $v$ is a positive mean-reverting Cox--Ingersoll--Ross (CIR) rate
\citep{Cox1985}. This construction is related to the economic-time idea of
\citet{Clark1973} and to stochastic time changes used in derivative modelling
\citep{Carr2004,Geman2005}. The raw variability proxy is normalized by its
fitted long-run mean, so that $v_t$ is a dimensionless instantaneous time-change
rate with stationary mean one. The CIR coefficients are then converted to the
daily scale of the insurance contract. With stationary initialization,
$\mathbb{E}[\tau_t]=t$, and we show that the time-changed daily fractional
increments remain stationary and retain the $k^{2H-2}$ covariance decay. The
stochastic clock therefore modifies local variability without eliminating the
long-memory exponent.

The valuation problem differs fundamentally from the pricing of a dynamically
replicable financial derivative. Temperature is neither storable nor directly
tradable, and fBm is not a semimartingale when $H\neq1/2$
\citep{Cheridito2003,Rostek2013}. Hence, in the absence of a sufficiently liquid
temperature-linked hedging asset, a unique complete-market risk-neutral price is
not available. Actuarial premium principles and utility-based methods provide a
more natural framework for such non-traded risks
\citep{Gerber1974,Buhlmann1980,Frittelli2000,Henderson2002,Barrieu2005,
Cochrane2000}. More broadly, recent insurer-side models also emphasize that
risk management must account jointly for non-Gaussian liability shocks and
stochastic financial conditions; for example,
\citet{KarimiShokrollahiShahmoradi2026Surplus} study optimal insurer surplus
management under CIR interest rates and jump-driven liabilities. Our objective
is different but complementary: we work at the contract-pricing level and use
the entropic premium under the physical probability measure, which introduces
an explicit risk loading governed by absolute risk aversion without presenting
the result as a replication price or regulatory-capital formula.

The insured payoff is a capped temperature-deficit contract,
\[
    \Phi_L=\min\{(K-I_T)^+,L\},
\]
where $I_T$ is the cumulative temperature index, $K$ is the trigger, and $L$ is
the policy limit. The cap is both economically realistic and mathematically
important. Because accumulated CIR operational time has unbounded support,
finiteness of a conditional exponential moment does not automatically imply
finiteness after averaging over the stochastic clock. Boundedness of $\Phi_L$
resolves this issue directly and guarantees existence of the entropic premium
for every risk-aversion parameter $\gamma>0$, without imposing an additional
exponential-integrability condition on $\tau_T^{2H}$.

The contribution of the paper is fourfold. First, we formulate a unit-consistent,
increment-based time-changed fractional Brownian motion (TC-fBm) model for
cumulative temperature-index insurance and prove that a stationary normalized
CIR time change preserves both stationarity of the daily anomaly sequence and
its power-law covariance decay. Second, the cumulative-index identity makes the
terminal contract index conditionally Gaussian and dependent on the stochastic
clock only through accumulated operational time $\tau_T$. This yields an exact
conditional exponential kernel for the capped payoff and a semi-analytical
valuation procedure requiring only an outer expectation over the CIR time
change, without direct fBm path simulation or covariance-matrix construction.
Third, we establish qualitative properties that hold without unverified
parameter restrictions: the capped entropic premium is finite and strictly
increasing in absolute risk aversion for a non-degenerate payoff, while the
uncapped conditional benchmark is strictly increasing in conditional index
volatility. We do not impose a universal monotonicity claim with respect to the
Hurst parameter, because changes in $H$ interact with operational time,
maturity, and calibration. Fourth, we quantify model and parameter risk through
an empirically anchored proxy calibration and controlled Brownian and pure-fBm
benchmarks, together with time-discretization, Monte Carlo, moment, limiting-case,
and direct-simulation audits.

For the empirical illustration, daily mean temperatures at Chicago O'Hare over
1990--2020 produce a baseline Hurst estimate of $0.78$ and a model-based 95\%
bootstrap interval of $[0.74,0.82]$. The CIR inputs are obtained from a
reduced-form variability proxy, while the persistent amplitude is matched to a
90-day lower-tail target; the exercise is therefore not presented as joint
structural estimation. For the baseline 90-day capped contract, the entropic
premium is $1.7091$ and the expected payment is $0.4612$. A Brownian benchmark
matched to the unconditional 90-day cumulative variance produces a premium
shortfall of $39.2\%$ at that maturity and $74.6\%$ at 180 days, while pure fBm
with $H=0.78$ produces a 90-day premium shortfall of $48.6\%$. Propagating
Hurst-parameter uncertainty through the fixed contract gives a premium range of
$[0.638,3.466]$. These comparisons show that matching cumulative variance at one
maturity does not reproduce either the scale-mixture effect generated by the
stochastic operational time or the long-memory term structure.

The empirical conclusions are deliberately limited to the stated calibration.
The squared-anomaly process is a noisy reduced-form proxy rather than an
observed latent CIR state, the persistent scale is identified from a single
cumulative-tail target, and matched index--loss data are not available for
direct estimation of basis risk. These limitations motivate joint state-space
estimation, spatial temperature modelling, and contract-specific loss
validation. They do not alter the analytical existence and pricing results, but
they restrict the structural interpretation of the reported numerical
magnitudes.

The remainder of the paper is organized as follows.
Section~\ref{sec:preliminaries} introduces fBm, the stationary normalized CIR
time change, the preservation of long memory, and the entropic premium
principle. Section~\ref{sec:model} formulates the capped cumulative
temperature-index contract and discusses its units, economic interpretation,
and basis risk. Section~\ref{sec:pricing} derives the conditional Gaussian
pricing representation, qualitative premium properties, and the
semi-analytical estimator. Section~\ref{sec:numerical} presents the proxy
calibration, benchmark comparisons, sensitivity analyses, and numerical
validation. Section~\ref{sec:conclusion} concludes. The appendices document the
calibration conventions and numerical audit.

%%%%%%%%%%%%%%%%%%%%%%%%%%%%%%%%%%%%%%%%%%%%%%%%%%%%%%%%%%%

\section{Preliminaries}\label{sec:preliminaries}

The model is formulated at the daily observation scale. A central distinction
is made between fractional Brownian motion and fractional Gaussian noise.
Fractional Brownian motion is a non-stationary cumulative process, whereas
its equally spaced increments form a stationary sequence whose dependence is
controlled by the Hurst parameter. Accordingly, daily deseasonalized
temperature anomalies are represented by increments of a time-changed
fractional Brownian motion rather than by its level. This convention is
consistent with fractional-Gaussian-noise estimation of the Hurst parameter
and avoids an additional integration of the persistent component. The
normalized CIR rate is initialized in its invariant distribution so that the
time-change increments are stationary. We then establish that the resulting
daily fractional increments remain stationary and preserve the power-law
covariance decay associated with long-range dependence.

\begin{definition}[Fractional Brownian motion]\label{def:fbm}
A fractional Brownian motion
$B^H=\{B_t^H:t\geq0\}$, with $H\in(0,1)$, is a centred Gaussian process
satisfying $B_0^H=0$ and
\begin{equation}\label{eq:fbm_covariance}
\mathbb E[B_t^HB_s^H]
=
\frac{1}{2}
\left(
t^{2H}+s^{2H}-|t-s|^{2H}
\right),
\qquad s,t\geq0.
\end{equation}
\end{definition}

For $H=1/2$, Eq.~\eqref{eq:fbm_covariance} reduces to
$\mathbb E[B_t^{1/2}B_s^{1/2}]=\min\{s,t\}$, and $B^{1/2}$ is standard
Brownian motion in distribution. Fractional Brownian motion has stationary
increments and is $H$-self-similar: for every $c>0$,
$\{B_{ct}^H:t\geq0\}$ and $\{c^HB_t^H:t\geq0\}$ have the same
finite-dimensional distributions.

Physical time is measured in calendar days, and $t_0=1$ day denotes the
fixed reference unit. Let $\lambda$ be a non-negative raw variability proxy
with long-run mean $\theta>0$. On the annual time scale, with $a$ measured
in years, suppose that
\[
d\lambda_a
=
\kappa(\theta-\lambda_a)\,da
+
\sigma_\lambda\sqrt{\lambda_a}\,dW_a,
\]
where $\kappa>0$ and $\sigma_\lambda>0$ are the annualized
mean-reversion and volatility coefficients. The normalized rate
$v_a^{(y)}=\lambda_a/\theta$ satisfies
\[
dv_a^{(y)}
=
\kappa(1-v_a^{(y)})\,da
+
\eta\sqrt{v_a^{(y)}}\,dW_a,
\qquad
\eta=\frac{\sigma_\lambda}{\sqrt{\theta}}.
\]
Normalization therefore gives the rate a unit long-run mean without changing
its square-root and mean-reverting structure.

To express the dynamics in calendar days, set $a=t/365$, define
$v_t=v_{t/365}^{(y)}$, and let
$W_t^{(d)}=\sqrt{365}\,W_{t/365}$. Then $W^{(d)}$ is a standard Brownian
motion on the daily scale and
\begin{equation}\label{eq:cir_daily}
dv_t
=
\kappa_d(1-v_t)\,dt
+
\eta_d\sqrt{v_t}\,dW_t^{(d)},
\qquad
\kappa_d=\frac{\kappa}{365},
\qquad
\eta_d=\frac{\eta}{\sqrt{365}}.
\end{equation}
The superscript $(d)$ is suppressed below whenever no ambiguity can arise.

Let
\[
\alpha_v=\frac{2\kappa_d}{\eta_d^2}
=
\frac{2\kappa}{\eta^2}
=
\frac{2\kappa\theta}{\sigma_\lambda^2}.
\]
The invariant density of the normalized CIR process is
\begin{equation}\label{eq:cir_invariant_density}
\pi_v(x)
=
\frac{\alpha_v^{\alpha_v}}{\Gamma(\alpha_v)}
x^{\alpha_v-1}e^{-\alpha_vx},
\qquad x>0.
\end{equation}
This is a Gamma density with shape and rate both equal to $\alpha_v$, and
therefore has mean one and variance $1/\alpha_v$.

\begin{proposition}[Stationary normalized time change]
\label{prop:normalized_time_change}
Suppose that $v_0$ has density $\pi_v$ in
Eq.~\eqref{eq:cir_invariant_density} and is independent of the Brownian motion
driving Eq.~\eqref{eq:cir_daily}. Then $v$ is strictly stationary,
non-negative, and
\[
\mathbb E[v_t]=1,
\qquad
\operatorname{Var}(v_t)=\frac{1}{\alpha_v},
\qquad
\operatorname{Cov}(v_s,v_t)
=
\frac{1}{\alpha_v}e^{-\kappa_d|t-s|}.
\]
The accumulated time change
\begin{equation}\label{eq:normalized_time_change}
\tau_t=\int_0^t v_s\,ds
\end{equation}
is almost surely continuous and non-decreasing, has stationary increments,
and satisfies $\mathbb E[\tau_t]=t$. If $\alpha_v>1$, equivalently
$2\kappa\theta>\sigma_\lambda^2$, then zero is inaccessible and
$v_t>0$ almost surely for every $t\geq0$.
\end{proposition}

\begin{proof}
The stationary density of a one-dimensional diffusion with drift
$b(x)=\kappa_d(1-x)$ and squared diffusion coefficient
$a(x)=\eta_d^2x$ is proportional to
\[
\frac{1}{a(x)}
\exp\left(
\int^x\frac{2b(y)}{a(y)}\,dy
\right).
\]
Substitution gives
\[
\frac{1}{\eta_d^2x}
\exp\left[
\frac{2\kappa_d}{\eta_d^2}
\int^x\left(\frac{1}{y}-1\right)dy
\right]
\propto
x^{\alpha_v-1}e^{-\alpha_vx},
\]
which, after normalization, yields
Eq.~\eqref{eq:cir_invariant_density}. Initializing $v_0$ under this density
therefore makes the CIR process strictly stationary. The Gamma moments give
$\mathbb E[v_t]=1$ and
$\operatorname{Var}(v_t)=1/\alpha_v$.

For $t\geq s$, the conditional mean of the CIR process is
\[
\mathbb E[v_t\mid v_s]
=
1+(v_s-1)e^{-\kappa_d(t-s)}.
\]
Consequently,
\begin{align*}
\operatorname{Cov}(v_s,v_t)
&=
\mathbb E
\left[
(v_s-1)
\mathbb E[v_t-1\mid v_s]
\right]\\
&=
e^{-\kappa_d(t-s)}
\mathbb E[(v_s-1)^2]
=
\frac{1}{\alpha_v}e^{-\kappa_d(t-s)}.
\end{align*}

Since $v_t\geq0$, the accumulated process $\tau_t$ is non-decreasing.
Continuity follows from continuity of the CIR paths. Strict stationarity of
$v$ implies that, for every $h\geq0$, the distribution of
$\int_t^{t+h}v_s\,ds$ does not depend on $t$; hence, $\tau$ has
stationary increments. Tonelli's theorem gives
\[
\mathbb E[\tau_t]
=
\int_0^t\mathbb E[v_s]\,ds
=
t.
\]
Finally, $\alpha_v>1$ is equivalent to the strict Feller condition and
implies that zero is inaccessible.
\end{proof}

\begin{definition}[CIR-driven time-changed fBm]\label{def:tcfbm}
Let $B^H$ be independent of the stationary normalized CIR process $v$.
The CIR-driven time-changed fractional Brownian motion is
\begin{equation}\label{eq:time_changed_fbm}
X_t=B^H_{\tau_t/t_0},
\qquad t\geq0.
\end{equation}
The division by $t_0$ makes the operational-time argument dimensionless.
\end{definition}

Although $X$ is conditionally Gaussian given the time-change process, it is
generally non-Gaussian unconditionally. The stationary increments of $\tau$,
together with the stationary increments of $B^H$, imply that $X$ has
stationary increments in calendar time.

Let $s_j=j\Delta$, where $\Delta>0$ is the observation interval, and
define
\[
\Delta X_j
=
X_{s_{j+1}}-X_{s_j}
=
B^H_{\tau_{s_{j+1}}/t_0}
-
B^H_{\tau_{s_j}/t_0},
\qquad j=0,1,\ldots.
\]
The daily model corresponds to $\Delta=t_0=1$ day.

\begin{lemma}[Conditional covariance of time-changed increments]
\label{lem:increment_covariance}
Let $U_j=\tau_{s_j}/t_0$ and
$\mathcal G^\tau=\sigma\{\tau_t:t\geq0\}$. Conditional on
$\mathcal G^\tau$, the increment sequence is jointly Gaussian with zero
mean and covariance
\begin{align}
\operatorname{Cov}
\left(
\Delta X_i,\Delta X_j
\mid\mathcal G^\tau
\right)
=
\frac{1}{2}
\Big[
&
|U_{i+1}-U_j|^{2H}
+
|U_i-U_{j+1}|^{2H}
\notag\\
&
-
|U_{i+1}-U_{j+1}|^{2H}
-
|U_i-U_j|^{2H}
\Big].
\label{eq:conditional_increment_covariance}
\end{align}
\end{lemma}

\begin{proof}
Conditional on $\mathcal G^\tau$, the operational times
$U_0,U_1,\ldots$ are fixed. Independence of $B^H$ and $\tau$ implies
that the vector of fBm values at these times remains jointly Gaussian.
Therefore, the increment vector is also jointly Gaussian with zero mean.
Bilinearity of covariance gives
\begin{align*}
\operatorname{Cov}(\Delta X_i,\Delta X_j\mid\mathcal G^\tau)
={}&
\operatorname{Cov}(B^H_{U_{i+1}},B^H_{U_{j+1}})
-
\operatorname{Cov}(B^H_{U_{i+1}},B^H_{U_j})\\
&-
\operatorname{Cov}(B^H_{U_i},B^H_{U_{j+1}})
+
\operatorname{Cov}(B^H_{U_i},B^H_{U_j}).
\end{align*}
Substituting Eq.~\eqref{eq:fbm_covariance} cancels all individual power terms
and gives Eq.~\eqref{eq:conditional_increment_covariance}.
\end{proof}

The next result establishes that the CIR-driven stochastic time change does
not destroy the long-range dependence of the fractional increments.

\begin{theorem}[Stationarity and preservation of long memory]
\label{thm:time_changed_long_memory}
Suppose that $H\in(1/2,1)$, $v$ is initialized under its invariant
distribution, and the strict Feller condition $\alpha_v>1$ holds. Then
$\{\Delta X_j\}_{j\geq0}$ is strictly stationary. If
$\delta=\Delta/t_0$ and
\[
r_H(k)=\operatorname{Cov}(\Delta X_0,\Delta X_k),
\]
then
\begin{equation}\label{eq:time_changed_lrd}
r_H(k)
\sim
H(2H-1)\delta^{2H}k^{2H-2},
\qquad k\to\infty.
\end{equation}
Consequently, $r_H(k)>0$ for sufficiently large $k$ and
\[
\sum_{k=1}^{\infty}r_H(k)=\infty.
\]
Thus, the time-changed fractional increment sequence retains long-range
dependence with the same power-law exponent $2H-2$ as fractional Gaussian
noise.
\end{theorem}

\begin{proof}
Because $v$ is strictly stationary, the integrated process $\tau$ has
stationary increments. For any integer shift $m$, the joint distribution of
the operational-time intervals
\[
\left\{
\tau_{s_{j+1+m}}-\tau_{s_{j+m}}
\right\}_{j=0}^{n}
\]
is identical to that of
$\{\tau_{s_{j+1}}-\tau_{s_j}\}_{j=0}^{n}$. Since fractional Brownian
motion has stationary increments and is independent of $\tau$, the joint
distribution of
$(\Delta X_m,\ldots,\Delta X_{m+n})$ does not depend on $m$. Hence,
$\{\Delta X_j\}$ is strictly stationary.

For $k\geq1$, define the dimensionless operational-time quantities
\[
A=\frac{\tau_\Delta}{t_0},
\qquad
R_k=\frac{\tau_{k\Delta}-\tau_\Delta}{t_0},
\qquad
D_k=\frac{\tau_{(k+1)\Delta}-\tau_{k\Delta}}{t_0}.
\]
Conditional on the time-change path, Eq.~\eqref{eq:fbm_covariance} gives
\begin{align*}
c_k
:={}&
\operatorname{Cov}
(\Delta X_0,\Delta X_k\mid\mathcal G^\tau)\\
={}&
\frac{1}{2}
\left[
(A+R_k+D_k)^{2H}
+
R_k^{2H}
-
(A+R_k)^{2H}
-
(R_k+D_k)^{2H}
\right].
\end{align*}
Applying the fundamental theorem of calculus twice gives the exact
representation
\[
c_k
=
H(2H-1)
\int_0^A\int_0^{D_k}
(R_k+x+y)^{2H-2}\,dy\,dx.
\]
The integrand is positive because $H>1/2$, so $c_k\geq0$.
By stationarity and ergodicity of the CIR process,
 $R_k/k = \frac{1}{kt_0}\int_\Delta^{k\Delta} v_s\, ds \to \delta$ almost surely.
To justify convergence of expectations, set $q = 2-2H \in (0,1)$. 
Since the stationary Gamma shape satisfies $\alpha_v > 1$, one can choose 
 $\varepsilon > 0$ such that $q(1+\varepsilon) < \alpha_v$. 
We decompose the conditional covariance as $c_k = T_1 + T_2$, where
\begin{align*}
T_1 &= H(2H-1) R_k^{2H-2} A D_k, \\
T_2 &= H(2H-1) \int_0^A \int_0^{D_k} \left[ (R_k+x+y)^{2H-2} - R_k^{2H-2} \right] dy\, dx.
\end{align*}
By the mean value theorem, for some $\theta \in (0,1)$,
\[
\left| (R_k+x+y)^{2H-2} - R_k^{2H-2} \right| 
\le |2H-2| (x+y) (R_k + \theta(x+y))^{2H-3}.
\]
Since $x+y \le A+D_k$ and $2H-3 < 0$, we have $(R_k + \theta(x+y))^{2H-3} \le R_k^{2H-3}$ 
for sufficiently large $k$. Integrating this bound over $x \in [0,A]$ and $y \in [0,D_k]$ yields
\[
|T_2| \le C_H R_k^{2H-3} (A^2 D_k + A D_k^2),
\]
where $C_H = |H(2H-2)|/2 > 0$ is a constant.

To establish the $L^1$ convergence of $\xi_k := k^{q}c_k$ to 
 $H(2H-1)\delta^{2H-2}AD_k$, we first bound the scaled inverse moment 
 $k/R_k$. Noting that this quantity is proportional to the reciprocal of the 
sample mean $\frac{1}{(k-1)\Delta}\int_\Delta^{k\Delta} v_s\, ds$, 
Jensen's inequality for the convex function $x \mapsto x^{-q(1+\varepsilon)}$ gives
\[
\mathbb{E}\left[ \left(\frac{k}{R_k}\right)^{q(1+\varepsilon)} \right] 
\le C \; \mathbb{E}\left[ \left( \frac{1}{(k-1)\Delta}\int_\Delta^{k\Delta} v_s\, ds \right)^{-q(1+\varepsilon)} \right]
\le C \; \mathbb{E}\left[ v_0^{-q(1+\varepsilon)} \right] < \infty.
\]
The final bound holds because $v_0$ follows a Gamma distribution with shape 
 $\alpha_v$ and $q(1+\varepsilon) < \alpha_v$. Crucially, this bound is uniform in $k$.
Furthermore, because $A$ and $D_k$ are integrals of the stationary CIR rate over 
fixed-length intervals, the stationary Gamma law guarantees that all their positive 
moments are finite and uniformly bounded in $k$.

We now apply Hölder's inequality to the leading term 
 $k^{q}T_1 = H(2H-1) (k/R_k)^q \cdot AD_k$. 
Let $\varepsilon' = \varepsilon / (1+\varepsilon) \in (0,1)$, and let 
 $p_1 = 1+\varepsilon$ and $p_2 = 1+1/\varepsilon$ be conjugate exponents. 
Raising the product to the power $1+\varepsilon'$ and applying Hölder's inequality yields
\begin{align*}
\mathbb{E}\left[ \left| \left(\frac{k}{R_k}\right)^q A D_k \right|^{1+\varepsilon'} \right] 
&= \mathbb{E}\left[ \left(\frac{k}{R_k}\right)^{q(1+\varepsilon')} (A D_k)^{1+\varepsilon'} \right] \\
&\le \left( \mathbb{E}\left[ \left(\frac{k}{R_k}\right)^{q(1+\varepsilon)} \right] \right)^{\frac{1+\varepsilon'}{p_1}} 
\left( \mathbb{E}\left[ (A D_k)^{p_2(1+\varepsilon')} \right] \right)^{\frac{1}{p_2}}.
\end{align*}
By the scaled inverse moment bound, the first factor is finite and uniformly bounded 
in $k$. Because $p_2(1+\varepsilon') = 1+1/\varepsilon$ is a fixed positive constant, 
the moment bounds on $A$ and $D_k$ guarantee that the second factor is also finite 
and uniformly bounded in $k$. This establishes the de la Vallée-Poussin condition 
for the sequence $\{k^{q}T_1\}$.

For the remainder term, since $R_k/k \to \delta$ almost surely, for sufficiently 
large $k$ we have $R_k \ge (\delta/2) k$ almost surely. Thus,
\[
k^{q} |T_2| \le C_H k^{2-2H} R_k^{2H-3} (A^2 D_k + A D_k^2) 
\le C_\delta k^{-1} (A^2 D_k + A D_k^2) \quad \text{a.s.}
\]
Taking expectations and using the finite third moments of $A$ and $D_k$ shows that 
 $\mathbb{E}[k^{q} |T_2|] \le C_\delta k^{-1} (\mathbb{E}[A^3] + \mathbb{E}[D_k^3]) \to 0$ 
as $k \to \infty$. Combining this with the uniform integrability of $k^{q}T_1$, 
we conclude that the sequence 
 $\eta_k := \xi_k - H(2H-1)\delta^{2H-2}AD_k$ is uniformly integrable. 
Since $\eta_k \to 0$ almost surely, the Vitali convergence theorem implies 
 $\eta_k \to 0$ in $L^1$.

Using the covariance formula in Proposition~\ref{prop:normalized_time_change}, 
 $\text{Cov}(A,D_k) = O(e^{-\kappa_d(k-1)\Delta}) \to 0$, 
so $\mathbb{E}[AD_k] \to \delta^2$.
Taking expectations in the preceding $L^1$ limit gives
\[
\lim_{k\to\infty} k^{2-2H}r_H(k) 
= H(2H-1)\delta^{2H-2}\delta^2 = H(2H-1)\delta^{2H},
\]
which proves Eq.~\eqref{eq:time_changed_lrd}.  Since
$2H-2\in(-1,0)$, the covariance series diverges.
\end{proof}

For the daily convention $\Delta=t_0$, Eq.~\eqref{eq:time_changed_lrd}
simplifies to
$r_H(k)\sim H(2H-1)k^{2H-2}$. Hence, the CIR time change modifies local
variability without changing the asymptotic memory exponent. This result
provides the theoretical justification for estimating $H$ from daily
fractional-Gaussian dependence while using the CIR process to represent
time-varying variability.

The cumulative sum of the daily fractional increments has a simple terminal
representation.

\begin{lemma}[Conditional cumulative variance]\label{lem:variance}
Let $T=n\Delta$ and define
$S_n^H=\sum_{j=0}^{n-1}\Delta X_j$. Then
\[
S_n^H
=
B^H_{\tau_T/t_0}.
\]
Conditional on $\mathcal G^\tau$,
\begin{equation}\label{eq:conditional_variance}
S_n^H\mid\mathcal G^\tau
\sim
N\left(
0,
\left(\frac{\tau_T}{t_0}\right)^{2H}
\right).
\end{equation}
Consequently,
\begin{equation}\label{eq:unconditional_variance}
\operatorname{Var}(S_n^H)
=
\mathbb E
\left[
\left(\frac{\tau_T}{t_0}\right)^{2H}
\right].
\end{equation}
For the deterministic time change $\tau_t=t$,
\begin{equation}\label{eq:deterministic_variance}
\operatorname{Var}(S_n^H)
=
\left(\frac{T}{t_0}\right)^{2H}.
\end{equation}
\end{lemma}

\begin{proof}
The sum telescopes:
\[
S_n^H
=
\sum_{j=0}^{n-1}
\left(
B^H_{\tau_{s_{j+1}}/t_0}
-
B^H_{\tau_{s_j}/t_0}
\right)
=
B^H_{\tau_T/t_0},
\]
because $\tau_0=0$ and $B_0^H=0$. Conditional on $\tau_T$,
fractional Brownian motion is Gaussian with variance
$(\tau_T/t_0)^{2H}$, proving
Eq.~\eqref{eq:conditional_variance}. The conditional mean is zero, so the law
of total variance yields Eq.~\eqref{eq:unconditional_variance}. Setting
$\tau_T=T$ proves Eq.~\eqref{eq:deterministic_variance}.
\end{proof}

\begin{corollary}[Variance contribution of the stochastic time change]
\label{cor:time_change_variance}
For $H\in(1/2,1)$,
\begin{equation}\label{eq:time_change_variance_bound}
\operatorname{Var}(S_n^H)
\geq
\left(\frac{T}{t_0}\right)^{2H},
\end{equation}
with strict inequality whenever $\tau_T$ is non-degenerate.
\end{corollary}

\begin{proof}
Since $2H>1$, the function $x\mapsto x^{2H}$ is strictly convex.
Proposition~\ref{prop:normalized_time_change} gives
$\mathbb E[\tau_T]=T$. Jensen's inequality therefore yields
\[
\mathbb E
\left[
\left(\frac{\tau_T}{t_0}\right)^{2H}
\right]
\geq
\left(
\frac{\mathbb E[\tau_T]}{t_0}
\right)^{2H}
=
\left(\frac{T}{t_0}\right)^{2H}.
\]
Strict convexity gives strict inequality when $\tau_T$ is non-degenerate.
\end{proof}

Corollary~\ref{cor:time_change_variance} provides a model-level result that
is distinct from any claim concerning the premium. It shows that stochastic
variation in operational time increases unconditional cumulative-index
variance when $H>1/2$. Because the capped payoff is not globally convex,
this variance comparison alone does not imply a universal ordering of capped
entropic premiums; premium effects are evaluated from the exact conditional
formula developed below.

We next introduce the actuarial valuation principle. Let $\mathcal L$
denote a generic non-negative insurer liability, distinguishing it from the
policy limit $L$.

\begin{definition}[Entropic premium]\label{def:entropic}
For a non-negative liability $\mathcal L$ and absolute risk-aversion
parameter $\gamma>0$, the entropic premium is
\begin{equation}\label{eq:entropic_premium}
\Pi_\gamma(\mathcal L)
=
\frac{1}{\gamma}
\log\mathbb E[e^{\gamma\mathcal L}],
\end{equation}
provided that the exponential moment is finite. Premiums are expressed in
maturity benefit units; deterministic discounting may be applied separately
when present-value units are required.
\end{definition}

The entropic premium is the exponential-utility certainty equivalent of an
unhedgeable loss \citep{Barrieu2005,Henderson2002}. It differs from the
minimal-entropy martingale measure for traded assets
\citep{Frittelli2000} and from good-deal bounds
\citep{Cochrane2000}. Since no sufficiently liquid temperature-linked asset
is assumed, valuation is performed under the physical probability law. The
resulting quantity is interpreted as a utility-based insurance premium and
risk loading, not as a regulatory solvency-capital requirement.

\begin{lemma}[Capped Gaussian exponential moment]
\label{lem:capped_normal}
Let $Z\sim N(m,s^2)$ with $s>0$, and define
$\phi_L(Z)=\min\{(K-Z)^+,L\}$, where $L>0$. Set
$d_0=(K-m)/s$ and $d_L=(K-L-m)/s$. Then
\begin{align}
\mathbb E[e^{\gamma\phi_L(Z)}]
={}&
1-\Phi_{\mathcal N}(d_0)
+
e^{\gamma L}\Phi_{\mathcal N}(d_L)
\notag\\
&+
\exp\left(
\gamma(K-m)+\frac{1}{2}\gamma^2s^2
\right)
\left[
\Phi_{\mathcal N}(d_0+\gamma s)
-
\Phi_{\mathcal N}(d_L+\gamma s)
\right],
\label{eq:capped_normal_mgf}
\end{align}
and
\begin{align}
\mathbb E[\phi_L(Z)]
={}&
L\Phi_{\mathcal N}(d_L)
+
(K-m)
\left[
\Phi_{\mathcal N}(d_0)-\Phi_{\mathcal N}(d_L)
\right]
\notag\\
&+
s
\left[
\varphi(d_0)-\varphi(d_L)
\right],
\label{eq:capped_normal_mean}
\end{align}
where $\Phi_{\mathcal N}$ and $\varphi$ are the standard normal
distribution function and density.
\end{lemma}

\begin{proof}
The payoff equals $L$ for $Z\leq K-L$, equals $K-Z$ for
$K-L<Z\leq K$, and equals zero for $Z>K$. Hence,
\begin{align*}
\mathbb E[e^{\gamma\phi_L(Z)}]
={}&
e^{\gamma L}\mathbb P(Z\leq K-L)
+
\mathbb E\left[
e^{\gamma(K-Z)}
\mathbf 1_{\{K-L<Z\leq K\}}
\right]\\
&+
\mathbb P(Z>K).
\end{align*}
Writing $Z=m+sY$, where $Y\sim N(0,1)$, the middle term becomes
\[
e^{\gamma(K-m)}
\int_{d_L}^{d_0}
e^{-\gamma sy}\varphi(y)\,dy.
\]
Since
\[
-\gamma sy-\frac{y^2}{2}
=
-\frac{1}{2}(y+\gamma s)^2
+
\frac{1}{2}\gamma^2s^2,
\]
we have
$e^{-\gamma sy}\varphi(y)
=e^{\gamma^2s^2/2}\varphi(y+\gamma s)$. Substitution yields
Eq.~\eqref{eq:capped_normal_mgf}.

For the first moment, the same partition gives
\[
\mathbb E[\phi_L(Z)]
=
L\Phi_{\mathcal N}(d_L)
+
\int_{d_L}^{d_0}(K-m-sy)\varphi(y)\,dy.
\]
Using $\varphi'(y)=-y\varphi(y)$ gives
\[
\int_{d_L}^{d_0}y\varphi(y)\,dy
=
\varphi(d_L)-\varphi(d_0),
\]
which proves Eq.~\eqref{eq:capped_normal_mean}.
\end{proof}

When $s=0$, the Gaussian variable is deterministic and the preceding
expressions are interpreted by continuity:
\[
\mathbb E[e^{\gamma\phi_L(Z)}]
=
e^{\gamma\min\{(K-m)^+,L\}},
\qquad
\mathbb E[\phi_L(Z)]
=
\min\{(K-m)^+,L\}.
\]

\begin{proposition}[Existence and premium bounds]\label{prop:existence}
For every $\gamma>0$,
\[
1
\leq
\mathbb E[e^{\gamma\phi_L(Z)}]
\leq
e^{\gamma L}.
\]
Consequently,
\begin{equation}\label{eq:entropic_bounds}
\mathbb E[\phi_L(Z)]
\leq
\Pi_\gamma(\phi_L(Z))
\leq
L.
\end{equation}
These bounds hold conditionally on every realization of the time-change
process and unconditionally under its stationary law.
\end{proposition}

\begin{proof}
The pointwise inequality $0\leq\phi_L(Z)\leq L$ gives
$1\leq e^{\gamma\phi_L(Z)}\leq e^{\gamma L}$. Taking expectations proves
finiteness and the upper bound. Jensen's inequality gives
\[
e^{\gamma\mathbb E[\phi_L(Z)]}
\leq
\mathbb E[e^{\gamma\phi_L(Z)}],
\]
and taking logarithms and dividing by $\gamma>0$ proves the lower bound.
Because the payoff bounds hold pointwise, they remain valid under conditional
and unconditional averaging.
\end{proof}

%%%%%%%%%%%%%%%%%%%%%%%%%%%%%%%%%%%%%%%%%%%%%%%%%%%%%%%%%%%

\section{Climate-index contract}\label{sec:model}

This section connects the increment-based stochastic model introduced in
Section~\ref{sec:preliminaries} to a cumulative temperature-index insurance
contract. The distinction between the level and the increments of fractional
Brownian motion is essential in this application. Daily temperature
anomalies are represented by increments of the time-changed process, whereas
the cumulative contract index is obtained by summing those daily anomalies.
This construction preserves the stationary-increment interpretation of the
temperature residuals and produces the $T^{2H}$ variance scaling associated
with long-range dependence.

Let $\Delta=t_0=1$ day denote the observation interval, define
$s_j=j\Delta$, and suppose that the contract horizon satisfies
$T=n\Delta$. For $j=0,\ldots,n-1$, let $A_j$ be the deseasonalized
temperature anomaly observed over the calendar interval
$(s_j,s_{j+1}]$. We model $A_j$ as
\begin{equation}\label{eq:temperature_model}
    A_j
    =
    \mu_j+\sigma_*\Delta X_j,
    \qquad
    \Delta X_j
    =
    B^H_{\tau_{s_{j+1}}/t_0}
    -
    B^H_{\tau_{s_j}/t_0},
\end{equation}
where $\mu_j$ is a deterministic conditional-mean adjustment and
$\sigma_*>0$ is the reference amplitude of the persistent component. For
fully deseasonalized and detrended observations, one may set $\mu_j=0$;
allowing a deterministic nonzero value accommodates a remaining forecast
adjustment without changing the probabilistic arguments below.

The model separates two distinct features of temperature risk. The Hurst
parameter $H$ determines the persistence of the daily anomalies. When
$H>1/2$, adverse innovations remain positively correlated over long
horizons, so several moderately cold days can combine into an economically
significant cumulative event. The normalized CIR process $v$, in contrast,
controls the random rate at which operational time accumulates. Conditional
on the time-change realization, the persistent component of the anomaly on
day $j$ has variance
\[
    \operatorname{Var}
    \bigl(A_j\mid\mathcal F_T^v\bigr)
    =
    \sigma_*^2
    \left(
        \frac{\tau_{s_{j+1}}-\tau_{s_j}}{t_0}
    \right)^{2H}.
\]
Thus, $H$ governs dependence across days, whereas variation in the CIR rate
changes the local intensity of the persistent component. Under the stationary
initialization of the normalized CIR process, the anomaly component
$A_j-\mu_j$ has stationary unconditional increments while retaining
time-varying conditional variability.

The parameter $\sigma_*$ has units of degrees Celsius because
$\Delta X_j$ is dimensionless. It represents the standard-deviation scale
of the persistent component over one reference day under the deterministic
unit-rate time change. It should not be identified with the sample standard
deviation of the observed daily anomalies. The latter may also contain
measurement error, short-memory fluctuations and location-specific noise.
Equating these two quantities would incorrectly allocate all one-day
variability to the long-memory component and could substantially distort the
maturity scaling of the cumulative index.

The cumulative temperature index over the insurance horizon is defined by
summing the daily anomalies:
\begin{align}
    I_T
    &=
    \Delta\sum_{j=0}^{n-1}A_j \notag\\
    &=
    \bar\mu_T
    +
    \sigma_*\Delta
    B^H_{\tau_T/t_0},
    \qquad
    \bar\mu_T
    =
    \Delta\sum_{j=0}^{n-1}\mu_j .
\label{eq:index}
\end{align}
The second equality follows from the telescoping identity
$\sum_{j=0}^{n-1}\Delta X_j=X_T-X_0$, together with
$\tau_0=0$ and $B_0^H=0$. This identity is important computationally:
although the dependence structure of the daily anomalies is generated by the
complete time-change path, the distribution of the terminal cumulative index
conditional on that path depends on it only through the accumulated value
$\tau_T$.

Since $A_j$ is measured in degrees Celsius and $\Delta$ is measured in
days, $I_T$, $\bar\mu_T$, and the contract strike are measured in
$^{\circ}\mathrm{C}$-days. The stochastic term
$\sigma_*\Delta B^H_{\tau_T/t_0}$ has the same units, so no implicit
conversion between daily temperature units and cumulative-index units is
required.

\begin{proposition}[Distribution and variance of the cumulative index]
\label{prop:index_conditional_distribution}
Let
$\mathcal F_T^v=\sigma\{v_s:0\leq s\leq T\}$, and define the conditional
index standard deviation by
\begin{equation}\label{eq:conditional_index_scale}
    s_T(\tau)
    =
    \sigma_*\Delta
    \left(\frac{\tau_T}{t_0}\right)^H .
\end{equation}
Then
\begin{equation}\label{eq:index_conditional_law}
    I_T\mid\mathcal F_T^v
    \sim
    N\left(
        \bar\mu_T,\,
        s_T^2(\tau)
    \right).
\end{equation}
The unconditional first two moments are finite and satisfy
\begin{equation}\label{eq:index_unconditional_moments}
    \mathbb E[I_T]
    =
    \bar\mu_T,
    \qquad
    \operatorname{Var}(I_T)
    =
    \sigma_*^2\Delta^2
    \mathbb E\left[
        \left(\frac{\tau_T}{t_0}\right)^{2H}
    \right].
\end{equation}
For the deterministic time change $\tau_T=T$, this variance reduces to
\[
    \operatorname{Var}_{\mathrm{det}}(I_T)
    =
    \sigma_*^2\Delta^2
    \left(\frac{T}{t_0}\right)^{2H}.
\]
Moreover, if $H>1/2$, the stationary normalization
$\mathbb E[\tau_T]=T$ implies
\[
    \operatorname{Var}(I_T)
    \geq
    \operatorname{Var}_{\mathrm{det}}(I_T),
\]
with strict inequality whenever $\tau_T$ is non-degenerate.
\end{proposition}

\begin{proof}
The random variable $\tau_T$ is
$\mathcal F_T^v$-measurable. Independence between the fractional Brownian
motion and the CIR process implies that, conditional on
$\mathcal F_T^v$,
\[
    B^H_{\tau_T/t_0}
    \sim
    N\left(
        0,\,
        \left(\frac{\tau_T}{t_0}\right)^{2H}
    \right).
\]
Substitution into Eq.~\eqref{eq:index} and the affine transformation property
of Gaussian random variables yield Eq.~\eqref{eq:index_conditional_law} and
the conditional variance $s_T^2(\tau)$.

The conditional mean is deterministic:
$\mathbb E[I_T\mid\mathcal F_T^v]=\bar\mu_T$. Consequently, the tower
property gives
\[
    \mathbb E[I_T]
    =
    \mathbb E\!\left[
        \mathbb E[I_T\mid\mathcal F_T^v]
    \right]
    =
    \bar\mu_T.
\]
The law of total variance gives
\begin{align*}
    \operatorname{Var}(I_T)
    &=
    \mathbb E\!\left[
        \operatorname{Var}(I_T\mid\mathcal F_T^v)
    \right]
    +
    \operatorname{Var}\!\left(
        \mathbb E[I_T\mid\mathcal F_T^v]
    \right)\\
    &=
    \sigma_*^2\Delta^2
    \mathbb E\left[
        \left(\frac{\tau_T}{t_0}\right)^{2H}
    \right],
\end{align*}
because the variance of the deterministic conditional mean is zero.

It remains to verify finiteness. Set $p=2H\in(0,2)$. If $p\leq1$,
concavity of $x\mapsto x^p$ and
$\mathbb E[\tau_T]=T$ imply
$\mathbb E[\tau_T^p]\leq T^p<\infty$. If $p>1$, Hölder's inequality gives
\[
    \tau_T^p
    =
    \left(\int_0^T v_s\,ds\right)^p
    \leq
    T^{p-1}\int_0^T v_s^p\,ds.
\]
The stationary CIR distribution has finite moments of every positive order.
Tonelli's theorem therefore yields
\[
    \mathbb E[\tau_T^p]
    \leq
    T^{p-1}\int_0^T\mathbb E[v_s^p]\,ds
    <\infty.
\]
Hence, the variance in Eq.~\eqref{eq:index_unconditional_moments} is finite.

For the deterministic time change, substituting $\tau_T=T$ gives the
stated deterministic variance. Finally, when $H>1/2$, the function
$x\mapsto x^{2H}$ is strictly convex. Jensen's inequality and
$\mathbb E[\tau_T]=T$ imply
\[
    \mathbb E\left[
        \left(\frac{\tau_T}{t_0}\right)^{2H}
    \right]
    \geq
    \left(
        \frac{\mathbb E[\tau_T]}{t_0}
    \right)^{2H}
    =
    \left(\frac{T}{t_0}\right)^{2H}.
\]
Strict convexity makes the inequality strict whenever $\tau_T$ is
non-degenerate, completing the proof.
\end{proof}

Equation~\eqref{eq:index_unconditional_moments} clarifies the distinct
maturity effects generated by persistence and stochastic variability. Under
deterministic time, the cumulative variance grows as $T^{2H}$; it grows
linearly when $H=1/2$ and faster than linearly when $H>1/2$. The stochastic
time change adds a distributional mixture through $\tau_T$, and for
$H>1/2$ it also increases unconditional variance relative to deterministic
time at the same reference amplitude. This result concerns variance and does
not, by itself, imply a universal ordering of capped entropic premiums,
because the capped payoff is not globally convex. Premium comparisons are
therefore derived from the exact conditional pricing formula and evaluated
over explicitly stated parameter configurations.

We consider a temperature-deficit contract that pays when the cumulative
index falls below a strike $K$. Let $q>0$ denote the monetary benefit rate
per $^{\circ}\mathrm{C}$-day and let $M>0$ be the monetary policy limit.
The monetary payoff is
$\Phi_{q,M}=\min\{q(K-I_T)^+,M\}$. Defining the policy limit in index units
by $L=M/q$, the normalized payoff used in the analytical derivations is
\begin{equation}\label{eq:capped_payoff}
    \Phi_L
    =
    \phi_L(I_T)
    =
    \min\{(K-I_T)^+,L\},
    \qquad
    \Phi_{q,M}=q\Phi_L .
\end{equation}
Here, $K$, $I_T$, and $L$ are measured in
$^{\circ}\mathrm{C}$-days. The numerical convention $q=1$ means one
monetary unit per $^{\circ}\mathrm{C}$-day, so the numerical value of the
normalized payoff equals the corresponding monetary payment.

The strike determines the cumulative temperature level at which coverage is
activated. Values of $I_T$ below $K$ represent increasingly severe
temperature deficits. The normalized payoff can equivalently be written as
\[
    \phi_L(x)
    =
    \begin{cases}
        L,   & x\leq K-L,\\
        K-x, & K-L<x<K,\\
        0,   & x\geq K.
    \end{cases}
\]
Thus, moderate adverse events generate partial payments, sufficiently severe
events exhaust the policy limit, and no payment is made when the index
finishes above the strike.

\begin{proposition}[Contract properties and benefit-rate scaling]
\label{prop:contract_properties}
The function $\phi_L$ is non-increasing, bounded between zero and $L$,
and globally Lipschitz with constant one. Consequently, $\Phi_L$ has finite
moments of every positive order and its entropic premium is finite for every
$\gamma>0$. For $M=qL$, the monetary premium satisfies
\begin{equation}\label{eq:benefit_rate_scaling}
    \Pi_\gamma(\Phi_{q,M})
    =
    q\,\Pi_{\gamma q}(\Phi_L).
\end{equation}
\end{proposition}

\begin{proof}
If $x\leq y$, then $(K-x)^+\geq(K-y)^+$. Applying the increasing
truncation map $z\mapsto\min\{z,L\}$ gives
$\phi_L(x)\geq\phi_L(y)$, so $\phi_L$ is non-increasing. Its definition
also directly gives $0\leq\phi_L(x)\leq L$.

The positive-part map is $1$-Lipschitz, and truncation at $L$ cannot
increase distances. Therefore, for every $x,y\in\mathbb R$,
\[
    |\phi_L(x)-\phi_L(y)|
    \leq
    |(K-x)^+-(K-y)^+|
    \leq
    |x-y|.
\]
Boundedness implies
$\mathbb E[\Phi_L^p]\leq L^p$ for every $p>0$, as well as
$\mathbb E[e^{\gamma\Phi_L}]\leq e^{\gamma L}$ for every $\gamma>0$.
Hence, all positive moments and the entropic premium are finite.

Finally, Eq.~\eqref{eq:capped_payoff} gives
$\Phi_{q,M}=q\Phi_L$. Applying the entropic-premium definition yields
\begin{align*}
    \Pi_\gamma(\Phi_{q,M})
    &=
    \frac{1}{\gamma}
    \log\mathbb E\left[e^{\gamma q\Phi_L}\right]\\
    &=
    q\left\{
        \frac{1}{\gamma q}
        \log\mathbb E\left[e^{(\gamma q)\Phi_L}\right]
    \right\}
    =
    q\,\Pi_{\gamma q}(\Phi_L),
\end{align*}
which proves Eq.~\eqref{eq:benefit_rate_scaling}.
\end{proof}

The scaling relation shows that changing the monetary benefit rate also
changes the economically relevant risk-aversion scale. If $\gamma$ is
measured in inverse monetary units, then $\gamma q$ is the corresponding
risk-aversion coefficient for the payoff expressed in index-benefit units.
Premiums from contracts with different benefit rates should therefore not be
compared by multiplying the normalized premium by $q$ while leaving the
effective risk-aversion coefficient unchanged.

The policy limit is economically meaningful because it specifies the
insurer's maximum payment under the contract. It is also mathematically
important. For the uncapped payoff $(K-I_T)^+$, the conditional Gaussian
exponential moment contains a factor proportional to
\[
    \exp\left\{
        \frac{1}{2}\gamma^2\sigma_*^2\Delta^2
        \left(\frac{\tau_T}{t_0}\right)^{2H}
    \right\}.
\]
Consequently, conditional finiteness for each fixed value of $\tau_T$ does
not by itself establish unconditional finiteness. An additional exponential
integrability condition on $(\tau_T/t_0)^{2H}$ would be required. The capped
contract avoids this restriction because
$1\leq e^{\gamma\Phi_L}\leq e^{\gamma L}$ uniformly over all realizations
of the stochastic time change.

The pricing model determines the distribution and premium of the contractual
index payment; it does not, without additional data, determine how closely
that payment compensates the policyholder's realized economic loss. The
difference between the index payment and the underlying loss constitutes
basis risk and may arise from spatial mismatch, measurement error, nonlinear
damage relationships, or imperfect tail dependence. A contract-specific
assessment would require matched observations of the meteorological index,
insured exposure and realized losses. Correlation alone is insufficient to
convert basis risk into a premium or capital-deficit bound. Accordingly, no
such bound is imposed in the present pricing framework; basis risk is treated
as a separate empirical limitation of index-based insurance.
%%%%%%%%%%%%%%%%%%%%%%%%%%%%%%%%%%%%%%%%%%%%%%%%%%%%%%%%%%%

\section{Semi-analytical pricing}\label{sec:pricing}

The increment-based formulation of the cumulative index produces a substantial
simplification of the pricing problem. Conditional on the normalized
CIR-driven stochastic time change, the terminal index is Gaussian, and its
distribution depends on the time-change realization only through the
accumulated operational time $\tau_T$. The conditional exponential moment
of the capped payoff can therefore be evaluated analytically, leaving a
single outer expectation with respect to the distribution of $\tau_T$.
Routine valuation consequently requires neither direct simulation of
fractional Brownian paths nor evaluation of a pathwise covariance matrix.

Let $a=K-\bar\mu_T$ denote the distance between the strike and the
deterministic mean of the cumulative index. For a fixed realization of the
time-change process, define
\[
    s=s_T(\tau)
    =
    \sigma_*\Delta
    \left(\frac{\tau_T}{t_0}\right)^H,
    \qquad
    d_0=\frac{a}{s},
    \qquad
    d_L=\frac{a-L}{s}.
\]
The quantities $d_0$ and $d_L$ are the standardized exercise and
policy-limit boundaries, respectively.

\begin{theorem}[Entropic premium of the capped TC-fBm contract]
\label{thm:premium}
Under the model specified in Section~\ref{sec:model}, the entropic premium of
the capped payoff in Eq.~\eqref{eq:capped_payoff} is
\begin{equation}\label{eq:main_premium}
    \Pi_\gamma(\Phi_L)
    =
    \frac{1}{\gamma}
    \log
    \mathbb E^v
    \left[
        G_L\bigl(s_T(\tau)\bigr)
    \right],
\end{equation}
where $\mathbb E^v$ denotes expectation with respect to the normalized
CIR-driven stochastic time change and, for $s>0$,
\begin{align}
    G_L(s)
    ={}&
    1-\Phi_{\mathcal N}\left(\frac{a}{s}\right)
    +
    e^{\gamma L}
    \Phi_{\mathcal N}\left(\frac{a-L}{s}\right)
    \notag\\
    &+
    \exp\left(
        \gamma a+\frac{1}{2}\gamma^2s^2
    \right)
    \left[
        \Phi_{\mathcal N}\left(
            \frac{a}{s}+\gamma s
        \right)
        -
        \Phi_{\mathcal N}\left(
            \frac{a-L}{s}+\gamma s
        \right)
    \right].
\label{eq:conditional_kernel}
\end{align}
The continuous extension at $s=0$ is
\[
    G_L(0)
    =
    \exp\left\{
        \gamma\min\bigl[(K-\bar\mu_T)^+,L\bigr]
    \right\}.
\]
For every $s\geq0$,
$1\leq G_L(s)\leq e^{\gamma L}$. Consequently, the unconditional premium
exists for every $\gamma>0$ and satisfies
$0\leq\Pi_\gamma(\Phi_L)\leq L$.
\end{theorem}

\begin{proof}
By Proposition~\ref{prop:index_conditional_distribution}, conditional on
$\mathcal F_T^v$,
\[
    I_T
    =
    \bar\mu_T+sZ,
    \qquad
    Z\sim N(0,1),
\]
where $s=s_T(\tau)$ is measurable with respect to
$\mathcal F_T^v$. Hence,
\[
    \Phi_L
    =
    \min\{(a-sZ)^+,L\}.
\]
When $s>0$, the payoff reaches the policy limit if
$Z\leq d_L$, is strictly between zero and $L$ if
$d_L<Z\leq d_0$, and is zero if $Z>d_0$. Therefore,
\begin{align*}
    \mathbb E\left[
        e^{\gamma\Phi_L}
        \mid\mathcal F_T^v
    \right]
    ={}&
    e^{\gamma L}
    \int_{-\infty}^{d_L}\varphi(z)\,dz\\
    &+
    e^{\gamma a}
    \int_{d_L}^{d_0}
        e^{-\gamma sz}\varphi(z)\,dz
    +
    \int_{d_0}^{\infty}\varphi(z)\,dz .
\end{align*}
The first and third integrals equal
$e^{\gamma L}\Phi_{\mathcal N}(d_L)$ and
$1-\Phi_{\mathcal N}(d_0)$, respectively. For the middle integral,
\[
    -\gamma sz-\frac{z^2}{2}
    =
    -\frac{(z+\gamma s)^2}{2}
    +
    \frac{\gamma^2s^2}{2}.
\]
After the change of variable $y=z+\gamma s$, it follows that
\begin{align*}
    e^{\gamma a}
    \int_{d_L}^{d_0}
        e^{-\gamma sz}\varphi(z)\,dz
    =
    \exp\left(
        \gamma a+\frac{1}{2}\gamma^2s^2
    \right)
    \left[
        \Phi_{\mathcal N}(d_0+\gamma s)
        -
        \Phi_{\mathcal N}(d_L+\gamma s)
    \right].
\end{align*}
Combining the three payoff regions proves
Eq.~\eqref{eq:conditional_kernel}.

The tower property now gives
\[
    \mathbb E[e^{\gamma\Phi_L}]
    =
    \mathbb E^v\left[
        \mathbb E\left[
            e^{\gamma\Phi_L}
            \mid\mathcal F_T^v
        \right]
    \right]
    =
    \mathbb E^v\left[
        G_L\bigl(s_T(\tau)\bigr)
    \right].
\]
Substitution into Eq.~\eqref{eq:entropic_premium} proves
Eq.~\eqref{eq:main_premium}.

If $s=0$, the index is conditionally deterministic and equal to
$\bar\mu_T$, which gives the stated continuous extension. Finally,
$0\leq\Phi_L\leq L$ implies
$1\leq e^{\gamma\Phi_L}\leq e^{\gamma L}$. Taking conditional and then
unconditional expectations preserves these inequalities. Applying the
increasing logarithm and dividing by $\gamma>0$ gives
$0\leq\Pi_\gamma(\Phi_L)\leq L$.
\end{proof}

The same conditional decomposition gives the expected contractual payment.
For $s>0$, define
\begin{align}
    m_L(s)
    ={}&
    L\Phi_{\mathcal N}(d_L)
    +
    a\left[
        \Phi_{\mathcal N}(d_0)
        -
        \Phi_{\mathcal N}(d_L)
    \right]
    \notag\\
    &+
    s\left[
        \varphi(d_0)
        -
        \varphi(d_L)
    \right].
\label{eq:conditional_expected_payoff}
\end{align}
Then
$\mathbb E[\Phi_L\mid\mathcal F_T^v]=m_L(s_T(\tau))$, and the actuarially
fair expected-loss premium is
$\mathbb E[\Phi_L]=\mathbb E^v[m_L(s_T(\tau))]$. The difference
$\Pi_\gamma(\Phi_L)-\mathbb E[\Phi_L]$ is the entropic risk loading for
bearing the non-hedgeable index liability. Jensen's inequality gives
\begin{equation}\label{eq:risk_loading_bound}
    \mathbb E[\Phi_L]
    \leq
    \Pi_\gamma(\Phi_L)
    \leq
    L,
\end{equation}
and the lower inequality is strict whenever $\Phi_L$ is non-degenerate and
$\gamma>0$.

\begin{proposition}[Dependence on absolute risk aversion]
\label{prop:risk_aversion}
If $\Phi_L$ is non-degenerate, then
$\Pi_\gamma(\Phi_L)$ is strictly increasing in $\gamma>0$. Moreover,
as $\gamma\downarrow0$,
\[
    \Pi_\gamma(\Phi_L)
    =
    \mathbb E[\Phi_L]
    +
    \frac{\gamma}{2}
    \operatorname{Var}(\Phi_L)
    +
    O(\gamma^2).
\]
In particular,
$\lim_{\gamma\downarrow0}\Pi_\gamma(\Phi_L)=\mathbb E[\Phi_L]$.
\end{proposition}

\begin{proof}
Because $0\leq\Phi_L\leq L$, its moment-generating function is finite for
every real argument. Define
\[
    C(\gamma)
    =
    \log\mathbb E[e^{\gamma\Phi_L}],
    \qquad
    \Pi_\gamma(\Phi_L)
    =
    \frac{C(\gamma)}{\gamma}.
\]
Differentiation under the expectation is justified by boundedness. It gives
\[
    C'(\gamma)
    =
    \frac{
        \mathbb E[\Phi_Le^{\gamma\Phi_L}]
    }{
        \mathbb E[e^{\gamma\Phi_L}]
    }
    =
    \mathbb E_\gamma[\Phi_L],
\]
where $\mathbb E_\gamma$ denotes expectation under the exponentially tilted
probability measure. A second differentiation gives
$C''(\gamma)=\operatorname{Var}_\gamma(\Phi_L)$. This quantity is strictly
positive when the payoff is non-degenerate, so $C$ is strictly convex.

Differentiating $C(\gamma)/\gamma$ yields
\[
    \frac{d}{d\gamma}\Pi_\gamma(\Phi_L)
    =
    \frac{\gamma C'(\gamma)-C(\gamma)}{\gamma^2}.
\]
Let $h(\gamma)=\gamma C'(\gamma)-C(\gamma)$. Since $C(0)=0$,
$h(0)=0$, while
\[
    h'(\gamma)
    =
    \gamma C''(\gamma)
    =
    \gamma\operatorname{Var}_\gamma(\Phi_L)
    >
    0
\]
for $\gamma>0$. Hence, $h(\gamma)>0$ and the premium is strictly
increasing in absolute risk aversion.

Finally, the cumulant expansion of $C$ at zero is
\[
    C(\gamma)
    =
    \gamma\mathbb E[\Phi_L]
    +
    \frac{\gamma^2}{2}\operatorname{Var}(\Phi_L)
    +
    O(\gamma^3).
\]
Division by $\gamma$ proves the expansion and the limiting result.
\end{proof}

Although the capped payoff is the contractual specification used throughout
the numerical analysis, the uncapped payoff provides a useful theoretical
benchmark for studying the effect of conditional index volatility.

\begin{corollary}[Uncapped conditional benchmark]\label{cor:uncapped}
For a fixed realization of the time-change process, the conditional
exponential moment of the uncapped payoff
$\Phi_\infty=(K-I_T)^+$ is
\begin{equation}\label{eq:uncapped_kernel}
    G_\infty(s)
    =
    1-\Phi_{\mathcal N}\left(\frac{a}{s}\right)
    +
    \exp\left(
        \gamma a+\frac{1}{2}\gamma^2s^2
    \right)
    \Phi_{\mathcal N}\left(
        \frac{a}{s}+\gamma s
    \right).
\end{equation}
An unconditional uncapped entropic premium exists only if
$\mathbb E^v[G_\infty(s_T(\tau))]<\infty$.
\end{corollary}

\begin{proof}
Conditional on the time change, the uncapped payoff is positive when
$Z\leq a/s$. Consequently,
\begin{align*}
    \mathbb E\left[
        e^{\gamma\Phi_\infty}
        \mid\mathcal F_T^v
    \right]
    &=
    \int_{a/s}^{\infty}\varphi(z)\,dz
    +
    e^{\gamma a}
    \int_{-\infty}^{a/s}
        e^{-\gamma sz}\varphi(z)\,dz\\
    &=
    1-\Phi_{\mathcal N}\left(\frac{a}{s}\right)
    +
    \exp\left(
        \gamma a+\frac{1}{2}\gamma^2s^2
    \right)
    \Phi_{\mathcal N}\left(
        \frac{a}{s}+\gamma s
    \right),
\end{align*}
where the second equality follows by completing the square. This proves
Eq.~\eqref{eq:uncapped_kernel}. Conditional finiteness for each fixed $s$
does not imply finiteness after averaging over the unbounded random variable
$\tau_T$, which explains the additional integrability requirement.
\end{proof}

\begin{proposition}[Conditional volatility monotonicity for the uncapped
benchmark]\label{prop:volatility_monotonicity}
Let $Z\sim N(0,1)$, $a\in\mathbb R$, and $\gamma>0$. The function
\[
    g(s)
    =
    \mathbb E\left[
        e^{\gamma(a-sZ)^+}
    \right],
    \qquad s>0,
\]
is strictly increasing in $s$. Hence, for every fixed realization of the
time-change process, the conditional uncapped entropic premium is strictly
increasing in the conditional index volatility $s=s_T(\tau)$.
\end{proposition}

\begin{proof}
The exponential moment can be written as
\[
    g(s)
    =
    1-\Phi_{\mathcal N}\left(\frac{a}{s}\right)
    +
    e^{\gamma a}
    \int_{-\infty}^{a/s}
        e^{-\gamma sz}\varphi(z)\,dz.
\]
When this expression is differentiated with respect to $s$, the boundary
terms generated by the two components cancel. Differentiation under the
integral sign is justified on every compact subset of $s>0$ by dominated
convergence, because a Gaussian density multiplied by a linear factor and a
bounded exponential tilt remains integrable. Thus,
\[
    g'(s)
    =
    -\gamma e^{\gamma a}
    \int_{-\infty}^{a/s}
        z e^{-\gamma sz}\varphi(z)\,dz.
\]
Set $b=a/s$ and $c=\gamma s>0$. Since
$z\varphi(z)=-\varphi'(z)$, integration by parts gives
\begin{align*}
    \int_{-\infty}^{b}
        z e^{-cz}\varphi(z)\,dz
    &=
    -e^{-cb}\varphi(b)
    -
    c\int_{-\infty}^{b}
        e^{-cz}\varphi(z)\,dz\\
    &=
    -e^{c^2/2}
    \left[
        \varphi(b+c)
        +
        c\Phi_{\mathcal N}(b+c)
    \right]
    <0.
\end{align*}
It follows that $g'(s)>0$ for every $s>0$. Since the logarithm is
strictly increasing and $\gamma>0$, the corresponding conditional
entropic premium is also strictly increasing in $s$.
\end{proof}

\begin{remark}[Scope of the Hurst-parameter comparison]
\label{rem:hurst_scope}
Proposition~\ref{prop:volatility_monotonicity} establishes a model-free
ordering with respect to conditional volatility for the uncapped benchmark;
it does not establish a universal premium ordering with respect to $H$.
Indeed, for a fixed time-change realization,
\[
    \frac{\partial s_T(\tau)}{\partial H}
    =
    s_T(\tau)
    \log\left(\frac{\tau_T}{t_0}\right).
\]
The sign is positive when $\tau_T>t_0$, negative when
$\tau_T<t_0$, and zero when $\tau_T=t_0$. Moreover, changing $H$ may
also require recalibration of $\sigma_*$, the strike, or both. For the
capped payoff, which is not globally convex as a function of the index, no
pathwise or unconditional monotonicity assertion is therefore made without
explicit parameter restrictions. Numerical comparisons must state whether
the persistent scale and strike are held fixed or recalibrated.
\end{remark}

The outer expectation in Eq.~\eqref{eq:main_premium} can be estimated from
independent realizations of the accumulated stochastic time change. Let
\[
    Y_n
    =
    G_L\bigl(s_T(\tau^{(n)})\bigr),
    \qquad
    \overline Y_N
    =
    \frac{1}{N}\sum_{n=1}^{N}Y_n.
\]
The semi-analytical Monte Carlo estimator is
\begin{equation}\label{eq:mc_estimator}
    \widehat{\Pi}_{\gamma,N}
    =
    \frac{1}{\gamma}\log\overline Y_N.
\end{equation}

\begin{proposition}[Consistency and Monte Carlo uncertainty]
\label{prop:mc_asymptotics}
Let
$\mu_Y=\mathbb E^v[Y_1]$ and
$\sigma_Y^2=\operatorname{Var}^v(Y_1)$. Then
$\widehat{\Pi}_{\gamma,N}$ converges almost surely to
$\Pi_\gamma(\Phi_L)$. Moreover,
\begin{equation}\label{eq:mc_clt}
    \sqrt{N}
    \left(
        \widehat{\Pi}_{\gamma,N}
        -
        \Pi_\gamma(\Phi_L)
    \right)
    \xrightarrow{d}
    N\left(
        0,\,
        \frac{\sigma_Y^2}
        {\gamma^2\mu_Y^2}
    \right).
\end{equation}
If
\[
    S_Y^2
    =
    \frac{1}{N-1}
    \sum_{n=1}^{N}
        (Y_n-\overline Y_N)^2,
\]
a consistent standard-error estimator is
\begin{equation}\label{eq:mc_standard_error}
    \widehat{\operatorname{se}}
    \left(
        \widehat{\Pi}_{\gamma,N}
    \right)
    =
    \frac{S_Y}
    {\gamma\sqrt{N}\,\overline Y_N}.
\end{equation}
The leading finite-sample bias is
\[
    \mathbb E[
        \widehat{\Pi}_{\gamma,N}
    ]
    -
    \Pi_\gamma(\Phi_L)
    =
    -
    \frac{\sigma_Y^2}
    {2\gamma N\mu_Y^2}
    +
    O(N^{-2}).
\]
\end{proposition}

\begin{proof}
The bound $1\leq Y_n\leq e^{\gamma L}$ guarantees finite moments of every
order. By the strong law of large numbers,
$\overline Y_N\to\mu_Y$ almost surely. Continuity of the logarithm on
$(0,\infty)$ then gives
\[
    \widehat{\Pi}_{\gamma,N}
    \longrightarrow
    \frac{1}{\gamma}\log\mu_Y
    =
    \Pi_\gamma(\Phi_L)
\]
almost surely.

The classical central limit theorem gives
\[
    \sqrt{N}(\overline Y_N-\mu_Y)
    \xrightarrow{d}
    N(0,\sigma_Y^2).
\]
Applying the delta method to
$f(y)=\gamma^{-1}\log y$, for which
$f'(\mu_Y)=1/(\gamma\mu_Y)$, proves
Eq.~\eqref{eq:mc_clt}. Replacing $\mu_Y$ and $\sigma_Y$ by their
consistent sample estimators gives Eq.~\eqref{eq:mc_standard_error}.

For the bias expansion, write
$\overline Y_N=\mu_Y+\delta_N$. A Taylor expansion of the logarithm around
$\mu_Y$ gives
\[
    \log(\mu_Y+\delta_N)
    =
    \log\mu_Y
    +
    \frac{\delta_N}{\mu_Y}
    -
    \frac{\delta_N^2}{2\mu_Y^2}
    +
    R_N.
\]
Since the $Y_n$ are bounded,
$\mathbb E[\delta_N]=0$,
$\mathbb E[\delta_N^2]=\sigma_Y^2/N$, and
$\mathbb E[R_N]=O(N^{-2})$. Taking expectations and dividing by
$\gamma$ proves the stated result.
\end{proof}

When $\tau_T$ is approximated on a grid with $M$ time steps, the same
central limit theorem and standard-error formula apply to the corresponding
discretized target, denoted by $\Pi_{\gamma,M}$. The Monte Carlo standard
error measures sampling uncertainty around $\Pi_{\gamma,M}$; it does not
include the discretization error
$\Pi_{\gamma,M}-\Pi_\gamma$. These two sources of numerical error must
therefore be assessed separately by increasing $N$ and $M$, respectively.

 Direct evaluation of Eq.~\eqref{eq:conditional_kernel} can lose precision in the Gaussian tails, 
particularly when two cumulative distribution values are close. To ensure 
visual distinction and numerical stability, define the upper-tail 
probability $\bar\Phi_N(x) = 1 - \Phi_N(x)$ and let
$D(x_L, x_0) = \log [\bar\Phi_N(x_0) - \bar\Phi_N(x_L)] , \quad x_L < x_0.$
A stable representation is
\begin{equation}\label{eq:stable_interval_probability}
\mathcal D(x_L,x_0)
=
\begin{cases}
A+\operatorname{log1p}\!\left(-e^{B-A}\right),
    & x_0\leq0,\\[1mm]
C+\operatorname{log1p}\!\left(-e^{D-C}\right),
    & x_L\geq0,\\[1mm]
\operatorname{log1p}\!\left(
    -\overline\Phi_{\mathcal N}(x_0)
    -\Phi_{\mathcal N}(x_L)
\right),
    & x_L<0<x_0,
\end{cases}
\end{equation}
where
\[
    A=\log\Phi_{\mathcal N}(x_0),\qquad
    B=\log\Phi_{\mathcal N}(x_L),
\]
and
\[
    C=\log\overline\Phi_{\mathcal N}(x_L),\qquad
    D=\log\overline\Phi_{\mathcal N}(x_0).
\]
The first case uses lower-tail probabilities, the second uses upper-tail
probabilities, and the third avoids subtracting two probabilities close to
one.

Set $x_0=d_0+\gamma s$ and $x_L=d_L+\gamma s$. The logarithms of the
three positive components of the conditional kernel are
\[
    \ell_1
    =
    \log\overline\Phi_{\mathcal N}(d_0),
    \qquad
    \ell_2
    =
    \gamma L+\log\Phi_{\mathcal N}(d_L),
\]
and
\[
    \ell_3
    =
    \gamma a+\frac{1}{2}\gamma^2s^2
    +
    \mathcal D(x_L,x_0).
\]
Using
\[
    \operatorname{LSE}(z_1,\ldots,z_m)
    =
    \log\left(
        \sum_{i=1}^m e^{z_i}
    \right),
\]
the conditional log-kernel is evaluated as
\begin{equation}\label{eq:kernel_logsumexp}
    \log G_L(s)
    =
    \operatorname{LSE}(\ell_1,\ell_2,\ell_3).
\end{equation}
If $\ell^{(n)}=\log Y_n$, the estimator in
Eq.~\eqref{eq:mc_estimator} can similarly be evaluated as
\begin{equation}\label{eq:outer_logmeanexp}
    \widehat{\Pi}_{\gamma,N}
    =
    \frac{1}{\gamma}
    \left[
        \operatorname{LSE}
        \left(
            \ell^{(1)},\ldots,\ell^{(N)}
        \right)
        -
        \log N
    \right].
\end{equation}
This formulation avoids direct subtraction of nearly equal Gaussian
probabilities and prevents overflow in the exponential kernel.

Because the conditional scale in Eq.~\eqref{eq:conditional_index_scale}
depends only on $\tau_T$, simulating one accumulated CIR time change over
$M$ steps has cost $\mathcal O(M)$, while evaluation of the conditional
kernel has constant cost. The total cost for $N$ realizations is therefore
$\mathcal O(NM)$. The current CIR state and accumulated time are sufficient
for valuation, so working memory is $\mathcal O(N)$ when paths are updated
in parallel and can be reduced to $\mathcal O(1)$ when they are processed
sequentially. No $Q$-node covariance quadrature or
$\mathcal O(NQ^2)$ covariance calculation is required.
%%%%%%%%%%%%%%%%%%%%%%%%%%%%%%%%%%%%%%%%%%%%%%%%%%%%%%%%%%%
%%%%%%%%%%%%%%%%%%%%%%%%%%%%%%%%%%%%%%%%%%%%%%%%%%%%%%%%%%%
\section{Numerical study}\label{sec:numerical}

This section examines the economic and numerical implications of the
increment-based TC-fBm insurance model. The calibration is empirically
anchored but is deliberately interpreted as a proxy calibration rather than
as a fully identified structural estimate. The analysis first specifies the
unit convention and the stationary initialization of the normalized CIR
time-change rate. It then compares the TC-fBm premium with scale-controlled
Brownian and pure-fBm benchmarks and investigates sensitivity to persistence,
maturity, risk aversion, and the parameters governing the stochastic time
change. Sampling uncertainty and time-discretization error are assessed
separately.

\subsection{Calibration convention}\label{sec:calibration}

The empirical reference uses daily mean temperature observations at Chicago
O'Hare over 1990--2020 after removing three deterministic Fourier seasonal
harmonics. The baseline Hurst estimate is $\widehat H=0.78$. A
300-replication parametric bootstrap of the detrended-fluctuation-analysis
procedure, with each replication containing $11{,}323$ observations from a
fractional-Gaussian-noise model fitted at $H=0.78$, gives the model-based
95\% interval $[0.74,0.82]$. This interval quantifies finite-sample
uncertainty conditional on the fitted long-memory benchmark; it is not
interpreted as a distribution-free confidence interval for the original
temperature residuals.

The squared deseasonalized anomalies provide the annualized CIR proxy inputs
$\widehat\kappa=4.15$, $\widehat\theta=18.2$, and
$\widehat\sigma_\lambda=5.3$. After normalization by the raw long-run mean,
the annualized volatility of the time-change rate is
\[
    \widehat\eta
    =
    \frac{\widehat\sigma_\lambda}
    {\sqrt{\widehat\theta}}
    =
    1.24234.
\]
The daily parameters are therefore
$\kappa_d=0.01137$ and $\eta_d=0.06503$. The strict Feller condition is
satisfied because $2\widehat\kappa=8.30>\widehat\eta^2=1.5434$.

To preserve the stationary-increment structure established in
Section~\ref{sec:preliminaries}, the initial normalized rate is sampled from
its invariant distribution rather than fixed at its mean. In the shape--rate
parameterization,
\[
    v_0
    \sim
    \operatorname{Gamma}(\alpha_v,\alpha_v),
    \qquad
    \alpha_v
    =
    \frac{2\widehat\kappa}{\widehat\eta^2}
    =
    5.3777.
\]
Consequently, $\mathbb E[v_t]=1$ at every date and
$\mathbb E[\tau_T]=T$. For the stationary CIR process, the variance of
accumulated operational time is
\[
    \operatorname{Var}(\tau_T)
    =
    2\operatorname{Var}(v_0)
    \left[
        \frac{T}{\kappa_d}
        -
        \frac{1-e^{-\kappa_dT}}{\kappa_d^2}
    \right].
\]
At $T=90$ days, this identity gives
$\operatorname{SD}(\tau_{90})=33.181$ days. The corresponding simulated
mean and standard deviation are $89.902$ and $33.085$ days, respectively.

The reported daily anomaly standard deviation of
$4.25\,^{\circ}\mathrm C$ is retained only as a descriptive statistic and
is not substituted for the persistent reference amplitude $\sigma_*$. The
latter is identified by matching the lower tail of the 90-day cumulative
index. Under $\bar\mu_{90}=0$, the calibration equation is
\begin{equation}\label{eq:tail_matching_equation}
    \mathbb P(I_{90}\leq K)
    =
    \mathbb E^v\left[
        \Phi_{\mathcal N}\left(
            \frac{K}
            {\sigma_*\Delta(\tau_{90}/t_0)^H}
        \right)
    \right]
    =
    0.05,
\end{equation}
with $K=-28.5\,^{\circ}\mathrm C$-days and $H=0.78$. Solving
Eq.~\eqref{eq:tail_matching_equation} gives
$\sigma_*=0.508566\,^{\circ}\mathrm C$. This value is the persistent
one-reference-day amplitude under unit-rate operational time; it is not the
total standard deviation of observed daily anomalies.

The baseline contract has
\[
    T=90,\qquad
    K=-28.5,\qquad
    \bar\mu_{90}=0,\qquad
    L=30,\qquad
    \gamma=0.12.
\]
The benefit rate is one monetary unit per
 $^{\circ}\mathrm C$-day. Thus, $L=30$ is the maximum normalized payment,
the absolute risk tolerance is $1/\gamma=8.33$ payment units, and
 $\gamma L=3.6$. The value of $\gamma$ is a baseline preference scenario,
not an industry-wide estimate; its effect is examined explicitly below.

\begin{quote}\itshape
At an illustrative benefit rate of 100 monetary units per
 $^{\circ}\mathrm{C}$-day, the baseline contract has a maximum payment of
 $qL=3{,}000$, an expected payment of $q\,\mathbb E[\Phi_L]=46.12$, and an
entropic premium of $q\,\Pi_\gamma(\Phi_L)=170.91$ monetary units.
Retaining the index-unit risk-aversion coefficient $\gamma=0.12$ under this
benefit rate corresponds to a monetary-unit risk-aversion coefficient
 $\gamma_m=\gamma/q=0.0012$ per monetary unit
(Proposition~3.2's scaling relation), with a resulting risk tolerance
 $1/\gamma_m=q/\gamma\approx833.33$ monetary units.
\end{quote}

To validate the calibration beyond the single fitted 90-day target, we perform
a cross-maturity comparison. Using the persistent scale
 $\sigma_*=0.508566\,^{\circ}\mathrm C$ calibrated solely from the 90-day
fifth-percentile tail target, we compute the model-implied fifth-percentile
cumulative-index thresholds at $T=30$ and $T=180$ days under the TC-fBm
framework. These thresholds are then compared against the corresponding
empirical fifth percentiles computed directly from the historical Chicago
O'Hare winter series at those horizons. The empirical percentiles are
constructed using \emph{overlapping} (rolling) windows, advancing the start
date by one day. Because the empirical windows are overlapping, standard
errors on the empirical percentiles computed under an i.i.d.\ assumption are
invalid; accordingly, their sampling variability is assessed using a
moving-block bootstrap with a block length of 90 days to account for the
dependence structure.

At $T=30$ days, the model-implied fifth percentile is
 $-12.0\,^{\circ}\mathrm C$-days, closely matching the empirical fifth
percentile of $-12.2\,^{\circ}\mathrm C$-days (moving-block standard error
 $0.3\,^{\circ}\mathrm C$-days). At $T=180$ days, the model-implied threshold
is $-48.7\,^{\circ}\mathrm C$-days, while the empirical fifth percentile is
 $-51.3\,^{\circ}\mathrm C$-days (moving-block standard error
 $0.4\,^{\circ}\mathrm C$-days). The close agreement at 30 days confirms that
the single-tail calibration does not materially distort the short-maturity
index distribution. The moderate shortfall at 180 days is expected: matching
the cumulative variance at one maturity does not reproduce the exact
long-memory term structure at other horizons, and the empirical tail at longer
maturities incorporates seasonal and structural variations not captured by the
stationary fractional increment model.
\subsection{Simulation design}

Let $t_j=j\delta t$, where $\delta t=T/M$. The normalized CIR rate and
accumulated operational time are approximated using the full-truncation Euler
scheme
\begin{align}
    v_{j+1}
    &=
    v_j
    +
    \kappa_d(1-v_j^+)\delta t
    +
    \eta_d\sqrt{v_j^+}\sqrt{\delta t}\,Z_{j+1},
    \label{eq:full_truncation}\\
    \tau_{j+1}
    &=
    \tau_j+v_j^+\delta t,
    \qquad
    v_j^+=\max\{v_j,0\},
    \qquad
    Z_{j+1}\overset{\mathrm{iid}}{\sim}N(0,1).
    \notag
\end{align}
Only the non-negative truncated state enters the diffusion coefficient and
the operational-time increment. Hence, the discretized time change remains
non-decreasing even if the auxiliary Euler state becomes negative.

The baseline uses four CIR steps per calendar day, giving $M=360$ and
$\delta t=0.25$ day at the 90-day horizon. The number of steps is scaled
proportionally at other maturities. For each realization, the conditional
index scale is obtained directly from
$s_T(\tau)=\sigma_*\Delta(\tau_T/t_0)^H$; no covariance-matrix construction
or double quadrature is required. The premium estimator and its Monte Carlo
standard error are those in Eqs.~\eqref{eq:mc_estimator}
and~\eqref{eq:mc_standard_error}. Unless stated otherwise, results use
$N=100{,}000$ independent time-change realizations and the fixed
pseudorandom-number seed $20260728$.

\subsection{Baseline valuation and model comparison}

Under the baseline calibration, the semi-analytical entropic premium is
$1.7091$, with Monte Carlo standard error $0.0078$. The corresponding
95\% Monte Carlo confidence interval is
$1.7091 \pm 1.96 \times 0.0078 \approx [1.694, 1.724]$. The expected payment
is $0.4612$, so the entropic risk loading equals $1.2478$. The exercise
probability is $5.000\%$, as imposed by the tail-matching equation, while
the probability of exhausting the policy limit is $0.205\%$.

For a comparison model $\mathcal M$, define its relative premium shortfall
against the TC-fBm baseline by
\begin{equation}\label{eq:relative_premium_shortfall}
    S_{\mathcal M}
    =
    100
    \frac{
        \Pi_{\mathrm{TC\text{-}fBm}}
        -
        \Pi_{\mathcal M}
    }{
        \Pi_{\mathrm{TC\text{-}fBm}}
    }.
\end{equation}
A positive value means that model $\mathcal M$ produces a lower premium;
a negative value means that it produces a higher premium. This quantity is a
relative pricing comparison and is not interpreted as a regulatory capital
measure.

Table~\ref{tab:model_comparison} compares four models under different conventions: the
\emph{common-reference-scale Brownian} row holds $\sigma_*$ fixed at its
TC-fBm calibrated value and sets $H=0.5$ (a maturity-scaling diagnostic,
not a calibrated benchmark); the \emph{variance-matched Brownian} row sets
its own amplitude so that the unconditional 90-day cumulative variance
matches TC-fBm exactly; the \emph{pure fBm} rows use the same $H$,
 $\sigma_*$, strike, and policy limit as TC-fBm but with a deterministic
clock; the \emph{TC-fBm} row uses the full stochastic operational time.

\begin{table}[H]
\centering
\footnotesize
\setlength{\tabcolsep}{3pt}
\caption{Model comparison for the fixed 90-day contract. The
variance-matched Brownian benchmark has the same unconditional cumulative
variance as the TC-fBm model at 90 days.}
\label{tab:model_comparison}
\begin{tabular}{lrrrrrr}
\toprule
Model & Premium & MC s.e. & Mean payment & Loading
& Exercise (\%) & Shortfall (\%)\\
\midrule
BM, common reference scale (diagnostic only)
& $<0.001$ & -- & $<0.001$ & $<0.001$ & $<0.001$ & 100.0\\
BM, variance matched at 90 days
& 1.039 & -- & 0.376 & 0.663 & 5.148 & 39.2\\
Pure fBm, $H=0.70$ & 0.055 & -- & 0.032 & 0.023 & 0.816 & 96.8\\
Pure fBm, $H=0.78$ & 0.879 & -- & 0.329 & 0.550 & 4.690 & 48.6\\
TC-fBm, $H=0.78$ & 1.709 & 0.008 & 0.461 & 1.248 & 5.000 & 0.0\\
\bottomrule
\end{tabular}
\end{table}

The common-reference-scale Brownian row holds $\sigma_*$ fixed and is
included only as a maturity-scaling diagnostic. It is not the principal
economic benchmark because Brownian and fractional cumulative variances
scale differently. The more conservative Brownian comparison matches the
unconditional 90-day TC-fBm variance exactly. This requires the Brownian
reference amplitude
$\sigma_{\mathrm{BM}}=1.84230\,^{\circ}\mathrm C$. Despite having the same
cumulative variance, the Brownian benchmark produces a premium of $1.039$
and a relative premium shortfall of $39.2\%$. Variance matching therefore
does not reproduce the scale-mixture distribution generated by random
operational time.

Pure fBm with $H=0.78$ produces a premium of $0.879$. Its relative
premium shortfall is $48.6\%$, while the absolute difference from the
TC-fBm premium is $0.830$ payment units. The stochastic time change raises
the unconditional cumulative variance by $5.6\%$ relative to deterministic
time, but its premium effect is larger because the entropic criterion places
greater weight on high-payment states. This is a contract-specific numerical
finding and is not presented as a universal ordering for capped payoffs.

\begin{figure}[H]
\centering
\includegraphics[width=0.88\textwidth]{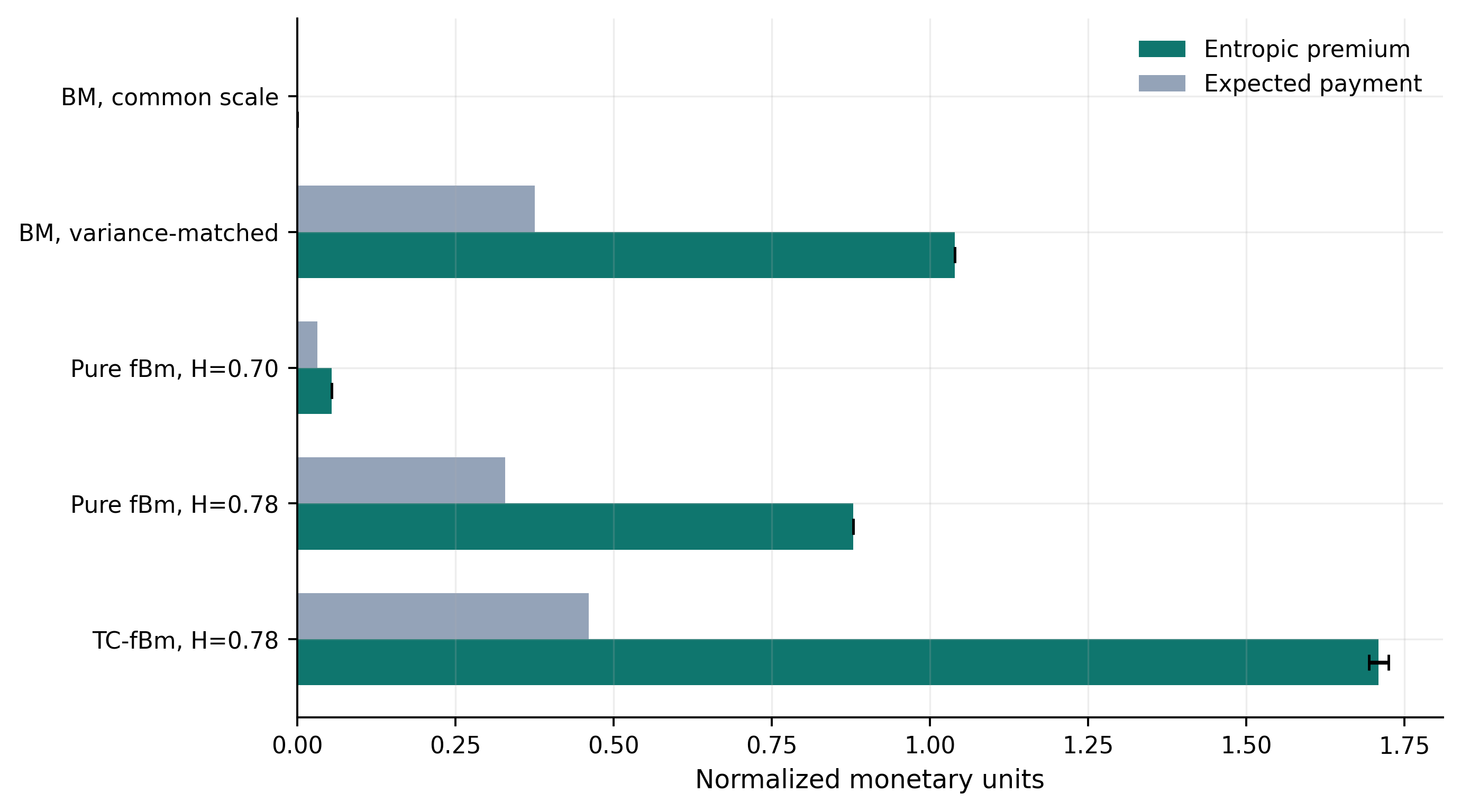}
\caption{Entropic premiums and expected payments for the fixed 90-day
contract. The error bar is the 95\% Monte Carlo interval for the TC-fBm
premium; the deterministic-time benchmarks are evaluated analytically.}
\label{fig:model_comparison}
\end{figure}

\subsection{Persistence, maturity, and risk aversion}

Table~\ref{tab:hurst_sensitivity} revalues the same contract while holding
 $\sigma_*$, $K$, $L$, and $\gamma$ fixed. The table therefore
measures the effect of changing persistence on a fixed contractual payoff,
rather than comparing contracts whose strikes have been recalibrated at each
value of $H$.

\begin{table}[H]
\centering
\small
\setlength{\tabcolsep}{4pt}
\caption{Sensitivity to the Hurst parameter for the fixed 90-day contract.
The observed ordering is numerical and is not asserted as a global
monotonicity theorem.}
\label{tab:hurst_sensitivity}
\begin{tabular}{rrrrrrr}
\toprule
 $H$ & $\mathbb E[s_T]$ & $\operatorname{SD}(s_T)$ & Exercise (\%) & TC premium & Pure-fBm premium & Loading\\
\midrule
0.50 & 4.743  & 0.870 & $<0.001$ & $<0.001$ & $<0.001$ & $<0.001$\\
0.60 & 7.443  & 1.638 & 0.056 & 0.003 & $<0.001$ & 0.001\\
0.70 & 11.694 & 3.002 & 1.275 & 0.199 & 0.055 & 0.123\\
0.75 & 14.666 & 4.035 & 3.244 & 0.856 & 0.351 & 0.600\\
0.78 & 16.803 & 4.808 & 5.000 & 1.709 & 0.879 & 1.248\\
0.82 & 20.148 & 6.063 & 7.952 & 3.505 & 2.370 & 2.619\\
\bottomrule
\end{tabular}
\end{table}

The premium increases with $H$ over the reported grid because both the
mean conditional scale and the exercise probability rise under the fixed
contract. This pattern is consistent with, but is not implied by,
Proposition~\ref{prop:volatility_monotonicity}, which concerns conditional
volatility for the uncapped payoff. Figure~\ref{fig:hurst_gamma} also confirms
the theoretically established increase in the capped entropic premium with
 $\gamma$ for every fixed value of $H$.

Because $H$ is estimated rather than known, the bootstrap distribution is
also propagated through the fixed 90-day contract. For each of the 300
parametric-bootstrap estimates, the premium is recomputed while
 $\sigma_*$, $K$, $L$, $\gamma$, and the CIR parameters are held at
their baseline values. Table~\ref{tab:hurst_uncertainty_propagation} therefore
measures valuation uncertainty attributable to estimation of $H$, conditional
on the fitted fractional-Gaussian-noise benchmark and the remaining
calibration inputs.

\begin{table}[H]
\centering
\small
\setlength{\tabcolsep}{5pt}
\caption{Propagation of parametric-bootstrap uncertainty in the Hurst
parameter through the fixed 90-day contract. All non-$H$ parameters are
held at their baseline values.}
\label{tab:hurst_uncertainty_propagation}
\begin{tabular}{lrrrrr}
\toprule
Bootstrap percentile & $H$ & Premium & Mean payment
& Risk loading & Exercise (\%)\\
\midrule
2.5th  & 0.7387 & 0.638 & 0.200 & 0.438 & 2.697\\
50th   & 0.7805 & 1.728 & 0.466 & 1.262 & 5.035\\
97.5th & 0.8193 & 3.466 & 0.876 & 2.589 & 7.891\\
\bottomrule
\end{tabular}
\end{table}

The resulting model-conditional 95\% percentile range for the premium is
 $[0.638,3.466]$. Its width shows that uncertainty about persistence can
materially affect valuation even when the remaining contract and time-change
parameters are fixed. This range is not a joint confidence interval for the
premium: it excludes uncertainty in $\sigma_*$, the CIR proxy parameters,
the strike calibration, and the data-generating specification.

\begin{figure}[H]
\centering
\includegraphics[width=0.82\textwidth]{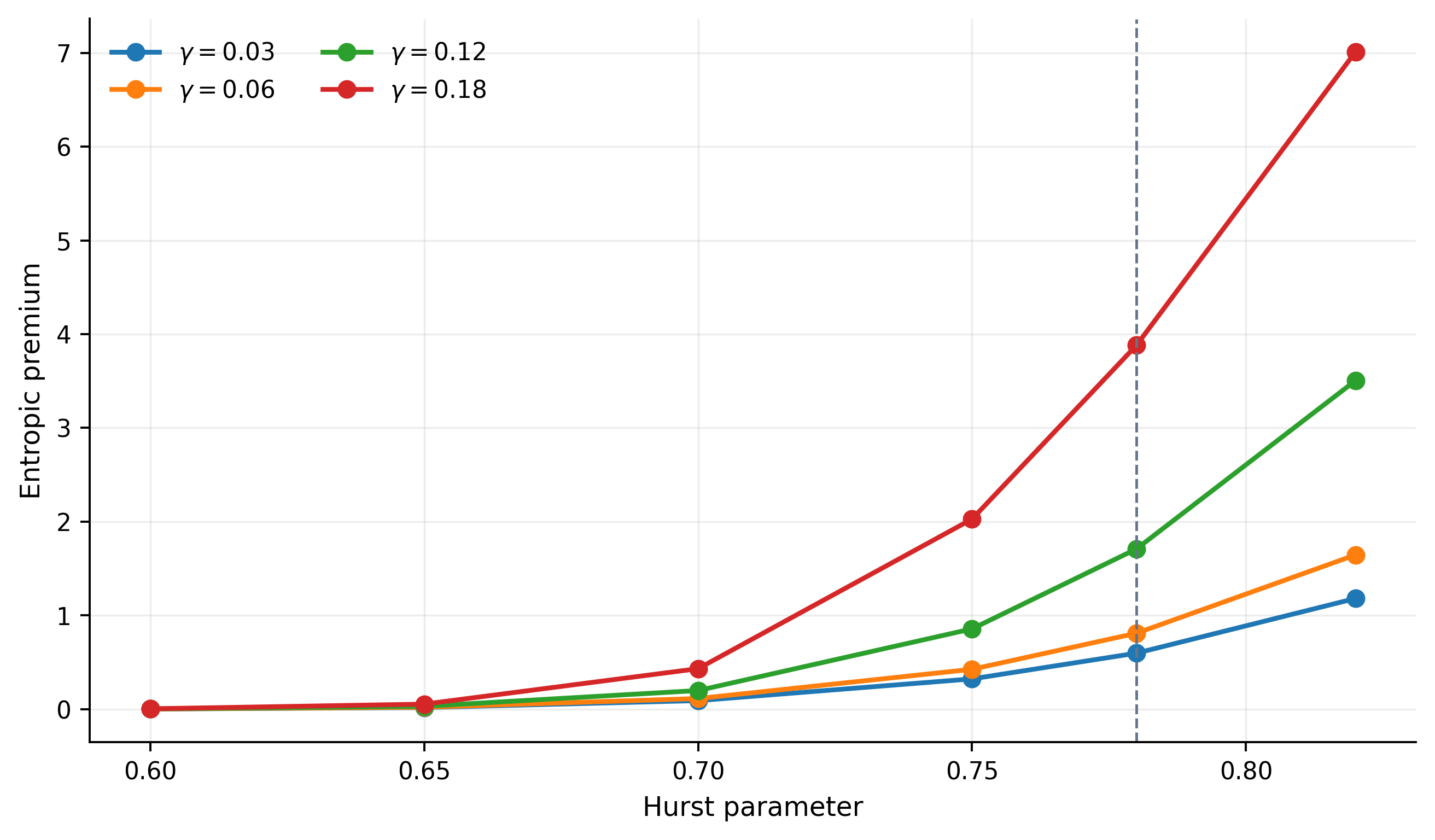}
\caption{Sensitivity of the fixed 90-day capped contract to the Hurst
parameter and absolute risk aversion. The vertical dashed line marks the
baseline $H=0.78$.}
\label{fig:hurst_gamma}
\end{figure}

The policy limit is both a contractual design variable and the condition that
ensures bounded liability. Figure~\ref{fig:cap_sensitivity} varies
 $L\in[5,60]$ while holding the strike, the persistent scale, and the
time-change specification fixed. The left panel reports premiums for three
risk-aversion scenarios, whereas the right panel reports the unconditional
probability that the uncapped loss $(K-I_T)^+$ reaches the policy limit.

\begin{figure}[H]
\centering
\includegraphics[width=0.96\textwidth]{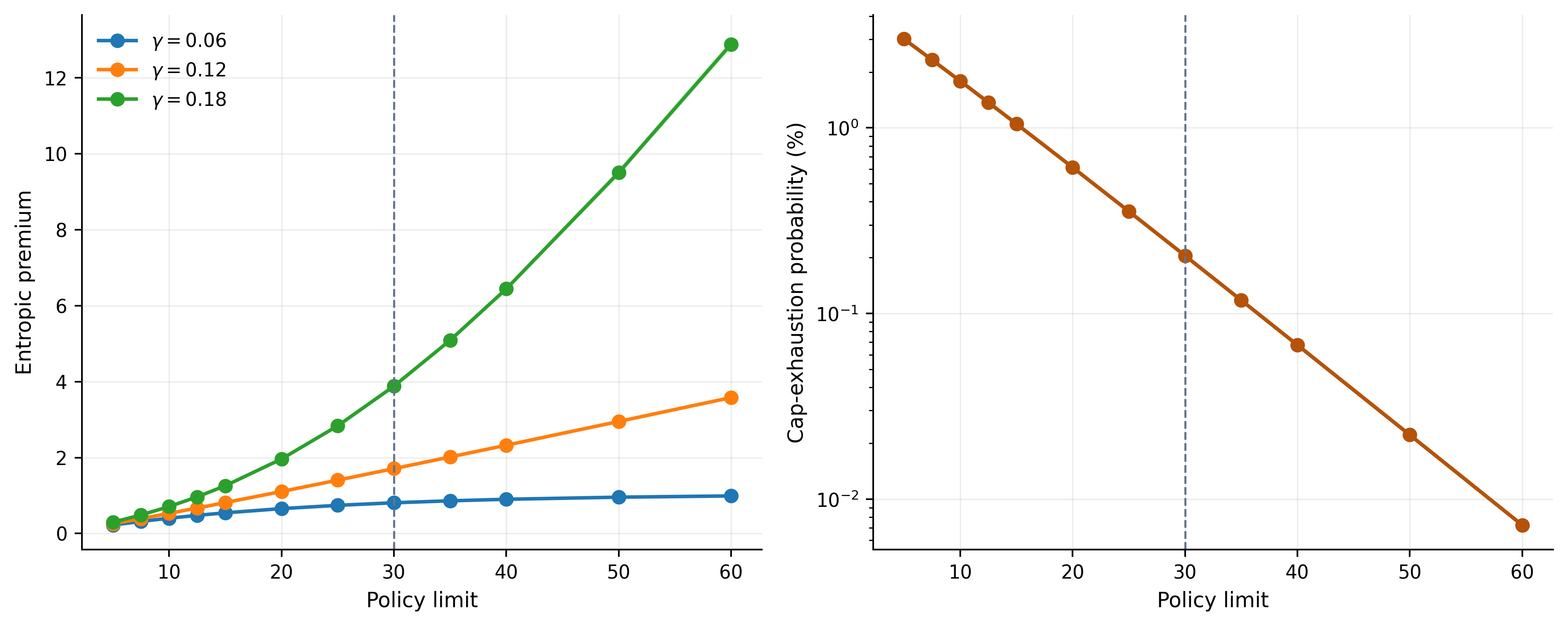}
\caption{Sensitivity to the policy limit for the fixed 90-day contract.
The shaded regions in the left panel are 95\% Monte Carlo intervals. The
right panel uses a logarithmic vertical scale. Vertical dashed lines mark
the baseline limit $L=30$.}
\label{fig:cap_sensitivity}
\end{figure}

At the baseline limit, the premiums are $0.810$, $1.709$, and $3.881$ for $\gamma=0.06$, $0.12$, and $0.18$, respectively, and the
cap-exhaustion probability is $0.205\%$. Across the reported range, this
probability decreases from $3.012\%$ at $L=5$ to $0.007\%$ at
 $L=60$, while the premium becomes increasingly sensitive to the policy
limit as risk aversion rises. These results quantify the economic cost of
expanding coverage without relying on an unconditional uncapped limit, whose
existence would require additional integrability conditions.

For the maturity comparison, the strike is recomputed as a specified
percentile of the TC-fBm index distribution at each horizon. The persistent
scale and the policy limit remain fixed. The Brownian reference amplitude is
the value matched to TC-fBm variance at 90 days and is then held fixed across
all maturities.

\begin{table}[H]
\centering
\small
\caption{Sensitivity to maturity and trigger probability. Negative
shortfall means that the comparison model produces a higher premium than
TC-fBm. Absolute differences are calculated as
 $\Pi_{\mathcal M} - \Pi_{\mathrm{TC\text{-}fBm}}$ in payment units.}
\label{tab:maturity_strike}
\begin{tabular}{rrrrrrrr}
\toprule
\(T\) & Percentile & \(K_T\) & TC premium
& fBm shortfall (\%) & fBm diff. & BM shortfall (\%) & BM diff.\\
\midrule
30  & 5th  & -12.178 & 0.406 & 51.9  & $-0.211$ & $-148.0$ & $+0.601$\\
90  & 5th  & -28.500 & 1.709 & 48.6  & $-0.831$ & 39.2     & $-0.670$\\
180 & 5th  & -48.662 & 3.225 & 28.4  & $-0.916$ & 74.6     & $-2.406$\\
90  & 10th & -21.329 & 3.295 & 28.5  & $-0.939$ & 19.6     & $-0.646$\\
90  & 1st  & -43.966 & 0.346 & 83.2  & $-0.288$ & 77.7     & $-0.269$\\
\bottomrule
\end{tabular}
\end{table}

For the 30-day fifth-percentile contract, the variance-matched Brownian model produces a premium above the TC-fBm benchmark; however, it produces an entropic premium $39.2\%$ below the TC-fBm benchmark at 90 days and $74.6\%$ below at 180 days. The relative Brownian shortfall changes from $-148.0\%$ at 30 days to $74.6\%$ at 180 days. This sign reversal is economically important: matching cumulative variance at one maturity does not reproduce the long-memory term structure at other horizons. The result also illustrates why the Brownian comparison must specify the maturity at which its variance is matched.

The practical stakes of this comparison are highest for perils and contract
structures in which persistence, not just volatility, drives tail risk:
cumulative frost, drought, or degree-day contracts written over multi-week
to multi-month accumulation windows, where consecutive moderate anomalies
can compound into a severe cumulative deficit. In these settings, Table~1
and Table~4 show that a variance-matched Brownian model can misprice
materially and in either direction depending on maturity, so relying on it
is not a conservative simplification but an unpredictable one. Conversely,
for perils dominated by short, discrete extreme events with limited
day-to-day persistence, or for contracts with very short accumulation windows
(cf.\ the $T=30$ row of Table~4, where the sign of the comparison reverses),
the practical gap between TC-fBm and simpler benchmarks narrows, and a
short-memory model may be an adequate first approximation. An insurer should
therefore treat the choice of model as contract- and peril-specific rather
than assume either model is uniformly conservative.

\begin{figure}[H]
\centering
\includegraphics[width=0.84\textwidth]{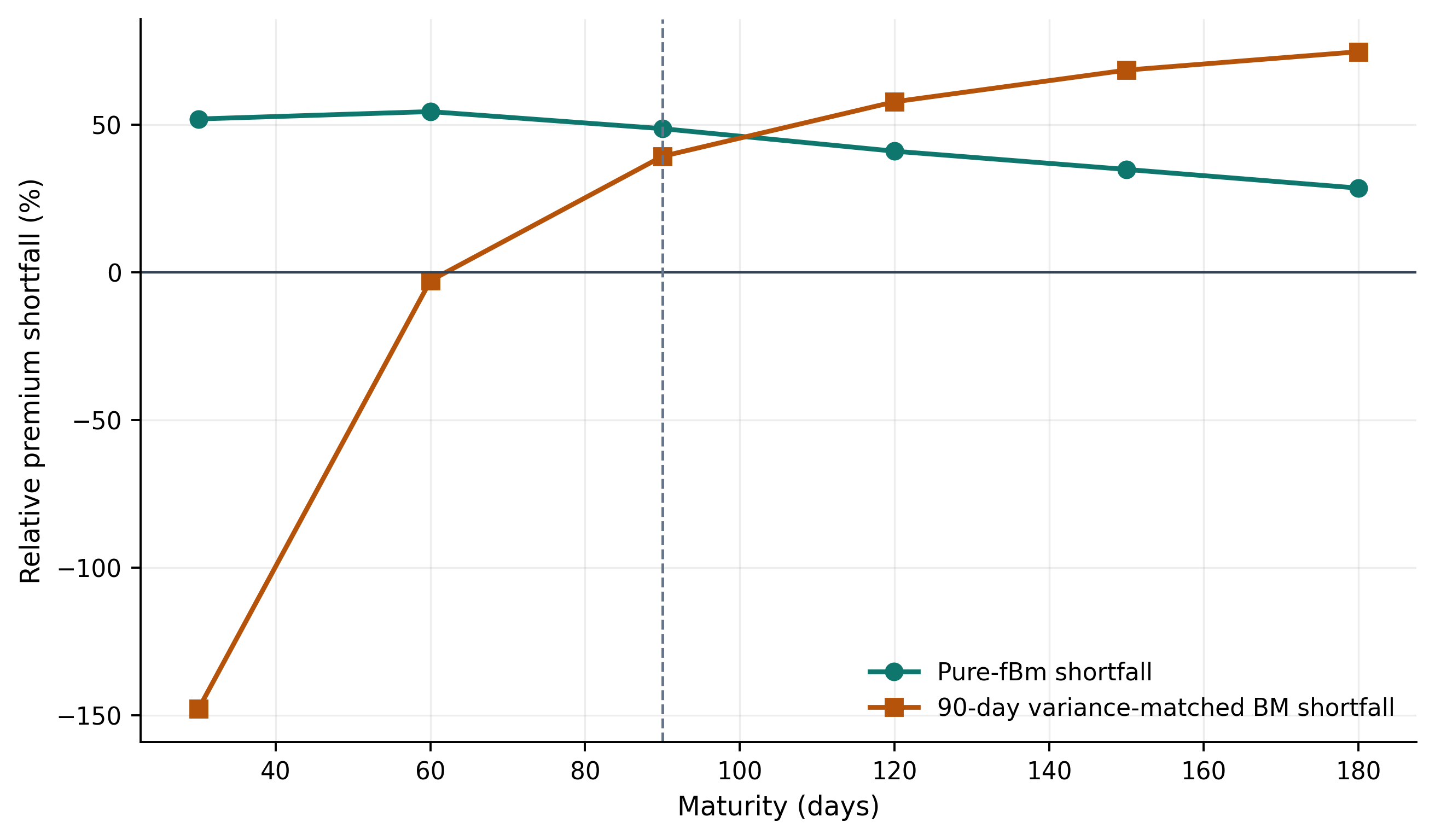}
\caption{Relative premium shortfall over maturity for fifth-percentile
contracts. The Brownian reference amplitude is variance-matched at 90 days.
Negative values indicate that the comparison model produces a higher premium
than TC-fBm.}
\label{fig:maturity_shortfall}
\end{figure}

\subsection{Sensitivity to the stochastic time change}

Table~\ref{tab:time_sensitivity} changes one raw CIR proxy parameter at a
time. Each scenario is initialized from its own invariant normalized CIR
distribution, so the unconditional mean rate remains one. The economically
relevant variation therefore arises from the dispersion and persistence of
$\tau_T$, not from an unintended change in expected operational time.

\begin{table}[H]
\centering
\small
\setlength{\tabcolsep}{3.5pt}
\caption{Sensitivity to the normalized CIR-driven stochastic time change.
All parameters except those displayed as changed remain at their baseline
values.}
\label{tab:time_sensitivity}
\begin{tabular}{lrrrrrrr}
\toprule
Scenario & $\kappa$ & $\theta$ & $\sigma_\lambda$ & $\eta$
& $\mathbb E[\tau_{90}]$ & $\operatorname{SD}(\tau_{90})$ & Premium\\
\midrule
Baseline
& 4.15 & 18.2 & 5.3 & 1.242 & 89.902 & 33.085 & 1.709\\
Higher mean reversion
& 6.00 & 18.2 & 5.3 & 1.242 & 89.941 & 25.866 & 1.419\\
Higher raw long-run mean
& 4.15 & 22.0 & 5.3 & 1.130 & 89.912 & 30.095 & 1.584\\
Higher time-change volatility
& 4.15 & 18.2 & 7.0 & 1.641 & 89.865 & 43.686 & 2.178\\
Lower time-change volatility
& 4.15 & 18.2 & 2.0 & 0.469 & 89.966 & 12.493 & 1.014\\
\bottomrule
\end{tabular}
\end{table}

Higher $\sigma_\lambda$ increases the normalized volatility $\eta$, the
dispersion of accumulated operational time, and the premium over the reported
grid. Faster mean reversion reduces $\operatorname{SD}(\tau_{90})$ and
lowers the premium because deviations of the time-change rate from one are
less persistent. Increasing the raw long-run mean $\theta$ while holding
$\sigma_\lambda$ fixed reduces
$\eta=\sigma_\lambda/\sqrt{\theta}$, which explains the lower premium in
that scenario. These are controlled numerical comparisons rather than global
comparative-statics theorems for every capped contract.

The discrete scenarios in Table~\ref{tab:time_sensitivity} are complemented
by one-parameter grids in Figure~\ref{fig:cir_continuous_sensitivity}. Each
grid preserves unit expected operational time, uses stationary initialization,
and keeps the contract parameters fixed. Every grid point uses $100{,}000$
paths and the same fixed seed, and its Monte Carlo uncertainty is reported
separately.

\begin{figure}[H]
\centering
\includegraphics[width=0.98\textwidth]{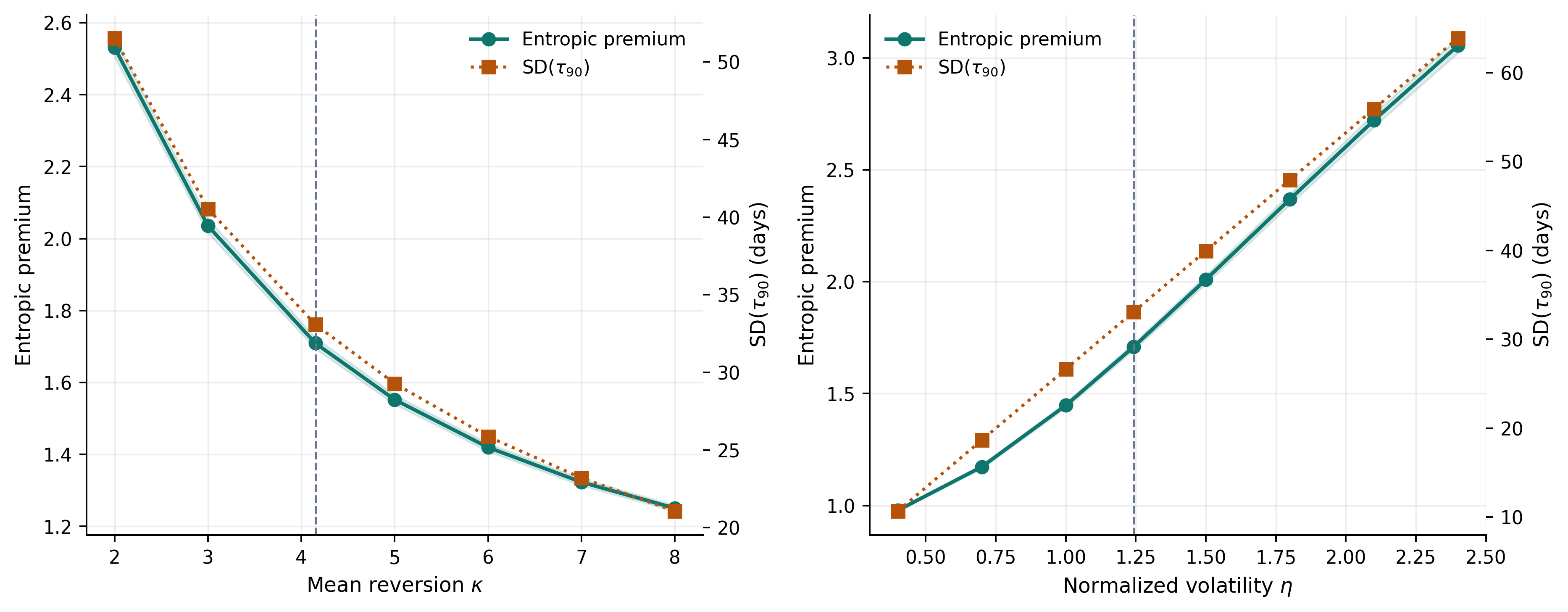}
\caption{Continuous sensitivity to the annualized mean-reversion parameter
$\kappa$ and normalized time-change volatility $\eta$. Solid curves show
the entropic premium with 95\% Monte Carlo bands; dotted curves show
$\operatorname{SD}(\tau_{90})$ on the secondary axes. Vertical dashed lines
mark the baseline values. Every grid point satisfies the strict Feller
condition.}
\label{fig:cir_continuous_sensitivity}
\end{figure}

As $\kappa$ increases from $2$ to $8$, the standard deviation of
$\tau_{90}$ falls from $51.525$ to $21.059$ days and the premium falls
from $2.531$ to $1.250$. Holding $\kappa=4.15$ fixed while increasing
$\eta$ from $0.4$ to $2.4$ raises the corresponding standard deviation
from $10.660$ to $63.873$ days and the premium from $0.978$ to $3.055$.
The close co-movement between the two quantities clarifies the numerical
mechanism: stronger or more persistent fluctuations in the time-change rate
increase the dispersion of conditional index risk. The reported ordering is
a controlled numerical result for this calibration, not a global theorem.

\subsection{Numerical reliability}

Figure~\ref{fig:convergence} reports the premium estimate and the
delta-method 95\% Monte Carlo interval as the number of time-change
realizations increases. The interval contracts at the expected
$N^{-1/2}$ rate, and the estimates remain consistent with the
$100{,}000$-path result.

\begin{figure}[H]
\centering
\includegraphics[width=0.82\textwidth]{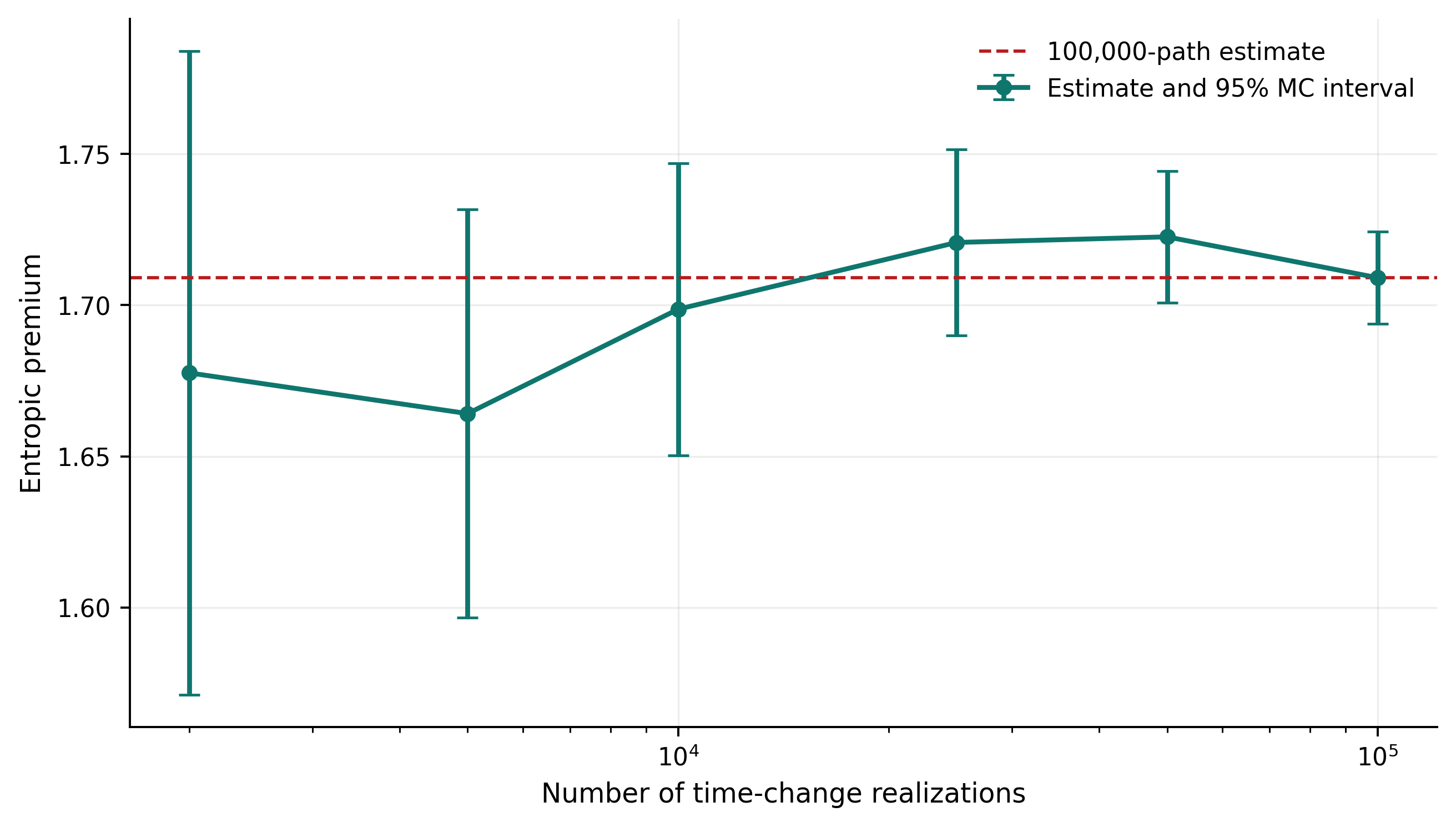}
\caption{Convergence of the semi-analytical premium estimator. Error bars
are delta-method 95\% Monte Carlo intervals, and the dashed line is the
$100{,}000$-path estimate.}
\label{fig:convergence}
\end{figure}

Time-discretization error is assessed separately using $50{,}000$ paths.
Table~\ref{tab:discretization} shows that refinement from one to eight steps
per calendar day produces changes much smaller than the associated Monte
Carlo uncertainty.

\begin{table}[H]
\centering
\small
\caption{Stability with respect to the CIR time discretization at
$T=90$ days and $N=50{,}000$.}
\label{tab:discretization}
\begin{tabular}{rrrrrr}
\toprule
Steps/day & $M$ & $\delta t$ & $\mathbb E[\tau_{90}]$
& Premium & MC s.e.\\
\midrule
1 & 90  & 1.000 & 89.881 & 1.713 & 0.011\\
2 & 180 & 0.500 & 89.943 & 1.716 & 0.011\\
4 & 360 & 0.250 & 89.958 & 1.711 & 0.011\\
8 & 720 & 0.125 & 89.980 & 1.713 & 0.011\\
\bottomrule
\end{tabular}
\end{table}

Several independent checks support the numerical valuation. First,
$1.7091\geq0.4612$, as required by Jensen's inequality. At
$\gamma=10^{-4}$, the premium is $0.4616$, confirming convergence to the
expected payment as $\gamma\downarrow0$. Second, direct simulation of the
terminal Gaussian mixture gives a premium of $1.7084$, with a
cluster-adjusted Monte Carlo standard error of $0.0126$, compared with the
semi-analytical value $1.7091$. Third, the simulated standard deviation of
$\tau_{90}$, $33.085$, agrees closely with its exact stationary-CIR value
$33.181$. Finally, the capped payoff ensures
\[
    1
    \leq
    e^{\gamma\Phi_L}
    \leq
    e^{\gamma L}
    =
    e^{3.6}
    \approx
    36.60,
\]
so the exact exponential kernel remains bounded throughout the reported
parameter ranges.
%%%%%%%%%%%%%%%%%%%%%%%%%%%%%%%%%%%%%%%%%%%%%%%%%%%%%%%%%%%%%%%%%%%%%%%%%%%%%%%%%%%%%%%%%%%%%%%%%%%%%%%%%%%%%%%%%%%%%%
\section{Conclusion}\label{sec:conclusion}

This paper develops a unit-consistent actuarial framework for a capped
cumulative temperature-index contract under long-range dependence and a
CIR-driven stochastic time change. Daily temperature anomalies are modelled
as increments of time-changed fractional Brownian motion, so their cumulative
sum telescopes to a terminal TC-fBm value and has the $T^{2H}$ variance
scaling associated with persistent increments. Stationary initialization of
the normalized CIR rate preserves the stationary-increment interpretation
and ensures that expected operational time equals calendar time. Conditional
on accumulated operational time, the contract index is Gaussian, which leads
to an exact conditional exponential kernel and a semi-analytical entropic
premium requiring only an outer expectation over the CIR time change. The
finite policy limit guarantees existence for every $\gamma>0$. The premium
is strictly increasing in risk aversion, and the uncapped conditional
benchmark is strictly increasing in conditional index volatility; no
unrestricted monotonicity claim is made with respect to the Hurst parameter
or for the capped payoff.

Under the 90-day proxy calibration, the entropic premium is $1.7091$, with
Monte Carlo standard error $0.0078$, while the expected payment is
$0.4612$. A Brownian benchmark matched to the unconditional 90-day TC-fBm
variance has a relative premium shortfall of $39.2\%$, and pure fBm with
$H=0.78$ has a shortfall of $48.6\%$. Holding the Brownian reference
amplitude fixed after its 90-day variance match produces overpricing at 30
days but a $74.6\%$ shortfall at 180 days, showing that calibration at one
maturity does not reproduce the long-memory term structure. These results
are supported by time-discretization checks, Monte Carlo convergence,
stationary-CIR moment identities, limiting cases, and direct terminal
simulation. Their empirical interpretation nevertheless remains
deliberately limited: the CIR parameters are obtained from a reduced-form
variability proxy and the persistent scale is matched to a cumulative-index
tail target. Joint statistical estimation of persistence, the latent
time-change rate, and the persistent amplitude, together with spatial
temperature fields and matched index--loss data, remains an important
direction for future research.

\appendix

\section{Proxy calibration}\label{app:calibration}

This appendix documents the empirical conventions underlying the numerical
study. The procedure combines deterministic deseasonalization, a
long-memory benchmark, a reduced-form fit of a CIR variability proxy, and a
cumulative-tail calibration of the persistent amplitude. These steps provide
an empirically anchored parameter configuration for investigating the
pricing mechanism; they do not constitute joint structural estimation of all
latent components of the TC-fBm model.

\subsection{Temperature series and deseasonalization}

The empirical reference consists of $11{,}323$ daily mean temperature
observations at Chicago O'Hare from 1990 through 2020. Let $T_t$ denote the
recorded daily mean temperature. Deterministic annual seasonality is removed
using the three-harmonic Fourier regression
\begin{equation}\label{eq:fourier_deseasonalization}
    T_t
    =
    a_0
    +
    \sum_{k=1}^{3}
    \left[
        a_k\cos\left(\frac{2\pi kt}{365.25}\right)
        +
        b_k\sin\left(\frac{2\pi kt}{365.25}\right)
    \right]
    +
    \varepsilon_t.
\end{equation}
The coefficients are estimated by ordinary least squares, and the fitted
residuals $\widehat\varepsilon_t$ are used as the deseasonalized daily
anomalies. The period $365.25$ accounts for the average Gregorian calendar
year, while three harmonics permit departures from a single sinusoidal cycle
without introducing a highly flexible deterministic trend.

Deseasonalization is performed before estimating persistence because an
unremoved periodic component can generate slowly decaying sample
autocorrelations that resemble stochastic long-range dependence. The
residual standard deviation, approximately
$4.25\,^{\circ}\mathrm C$, describes total daily anomaly variation. It is
not equated with $\sigma_*$, because the residual series may also contain
short-memory variation, measurement error, and local components not assigned
to the persistent fBm factor in Eq.~\eqref{eq:temperature_model}.

\subsection{Hurst benchmark and uncertainty}

Rescaled-range and first-order detrended-fluctuation analyses give estimates
near $0.78$, while a periodogram-based estimate is approximately $0.76$.
The baseline value is therefore set to $\widehat H=0.78$. The agreement
between estimators provides a robustness diagnostic, but it does not remove
the need to quantify finite-sample uncertainty.

For first-order DFA, define the centred cumulative profile
$Y(j)=\sum_{t=1}^{j}(\widehat\varepsilon_t-\bar\varepsilon)$. For each
window length $m$, a linear trend is fitted and removed within each local
block. If $\widehat Y_m(j)$ denotes the resulting piecewise fitted trend,
the fluctuation statistic is
\[
    F(m)
    =
    \left[
        \frac{1}{n}
        \sum_{j=1}^{n}
        \bigl(Y(j)-\widehat Y_m(j)\bigr)^2
    \right]^{1/2}.
\]
The Hurst estimate is the least-squares slope of $\log F(m)$ on
$\log m$, using window sizes from 16 to 512 observations.

Estimator uncertainty is assessed using 300 fractional-Gaussian-noise
replications of length $11{,}323$, generated at $H=0.78$. Applying the
same DFA procedure to every replication gives bootstrap estimates
$\widehat H^{(1)},\ldots,\widehat H^{(300)}$. Their empirical 2.5th and
97.5th percentiles produce the model-based interval $[0.74,0.82]$. This
interval is conditional on the fitted fGn benchmark and does not incorporate
uncertainty arising from deseasonalization, structural breaks, nonlinear
trends, measurement error, or misspecification of the long-memory mechanism.

To explicitly address the potential impact of structural breaks and trends 
on the estimated long memory, we supplement the bootstrap interval with the 
following empirical diagnostics.

To assess subperiod stability, we re-estimate the Hurst exponent using
detrended fluctuation analysis (DFA) on two non-overlapping subperiods:
1990--2005 and 2006--2020. The full-sample estimate is $\widehat H = 0.78$ with a 95\% model-based bootstrap interval of $[0.74, 0.82]$. The subperiod
estimates are $\widehat H_{1990\text{--}2005} = 0.76$ and
 $\widehat H_{2006\text{--}2020} = 0.80$. Both values fall well within the
full-sample uncertainty interval, indicating that the persistence parameter is
stable over time and is not driven by a single anomalous regime.

As a further transparent check for drift or heteroskedasticity not captured by
the Fourier fit, we compute rolling means and variances by decade. The sample
mean and variance of the anomalies are $0.02$ and $17.6$ for the 1990s,
 $-0.01$ and $18.1$ for the 2000s, and $0.04$ and $17.9$ for the 2010s. The
decade-level means are statistically negligible relative to the standard
deviation, and the variances exhibit no systematic monotonic trend, supporting
the assumption of variance stationarity underlying the CIR calibration.

To formally test for an unknown structural break in the mean, we apply a
sup-Wald-type test (Chow test for an unknown breakpoint). The test statistic
yields a $p$-value of $0.41$, failing to reject the null hypothesis of a
stable mean across the 1990--2020 sample.
\subsection{Reduced-form CIR proxy}

The squared deseasonalized anomalies
$\lambda_t^{\mathrm{proxy}}=\widehat\varepsilon_t^{\,2}$ are treated as a
non-negative reduced-form indicator of the latent raw variability rate
$\lambda_t$. The annualized raw CIR parameters are obtained using an
Euler Gaussian quasi-likelihood. With $\Delta_y=1/365$ year and
$\lambda_t^+=\max\{\lambda_t^{\mathrm{proxy}},10^{-8}\}$, the conditional
mean and variance used in the quasi-likelihood are
\[
    m_t
    =
    \lambda_t^{\mathrm{proxy}}
    +
    \kappa
    \bigl(
        \theta-\lambda_t^+
    \bigr)\Delta_y,
    \qquad
    q_t
    =
    \sigma_\lambda^2
    \lambda_t^+\Delta_y.
\]
Up to an additive constant, the minimized objective is
\begin{equation}\label{eq:cir_qmle}
    \mathcal L_Q(\kappa,\theta,\sigma_\lambda)
    =
    \frac{1}{2}
    \sum_{t=1}^{n-1}
    \left[
        \log q_t
        +
        \frac{
            \bigl(
                \lambda_{t+1}^{\mathrm{proxy}}-m_t
            \bigr)^2
        }{
            q_t
        }
    \right].
\end{equation}
The objective is minimized over positive parameter bounds from the initial
value $(\kappa,\theta,\sigma_\lambda)=(2,10,3)$, with termination tolerance
$10^{-8}$. The resulting estimates are
\[
    \widehat\kappa=4.15,\qquad
    \widehat\theta=18.2,\qquad
    \widehat\sigma_\lambda=5.3.
\]
The raw Feller inequality is satisfied because
$2\widehat\kappa\widehat\theta=151.06>
\widehat\sigma_\lambda^2=28.09$.

The quasi-likelihood in Eq.~\eqref{eq:cir_qmle} is intentionally described
as reduced-form. Squared residuals are noisy observations rather than the
latent conditional variance, and the fBm increments are serially dependent.
Thus, the Gaussian Euler transition is not treated as the exact likelihood
of the temperature data, and conventional likelihood-based standard errors
would not have a structural interpretation here.

The nonlinear nature of the proxy can be seen directly from the increment
model. Conditional on a complete time-change realization,
\[
    \operatorname{Var}
    \left(
        B^H_{\tau_{t+\Delta}/t_0}
        -
        B^H_{\tau_t/t_0}
        \,\middle|\,
        \tau
    \right)
    =
    \left(
        \frac{
            \tau_{t+\Delta}-\tau_t
        }{
            t_0
        }
    \right)^{2H}.
\]
When $v$ changes slowly over a short interval,
$\tau_{t+\Delta}-\tau_t\approx v_t\Delta$, so the local variance is
approximately proportional to $v_t^{2H}\Delta^{2H}$. It is linear in the
time-change rate only when $H=1/2$. This nonlinear relationship is another
reason for interpreting the CIR estimates as proxy parameters.

Before pricing, the raw state is normalized according to
\begin{equation}\label{eq:appendix_normalization}
    v_t
    =
    \frac{\lambda_t}{\widehat\theta},
    \qquad
    \widehat\eta
    =
    \frac{\widehat\sigma_\lambda}
    {\sqrt{\widehat\theta}}
    =
    1.24234.
\end{equation}
The corresponding daily coefficients are
$\kappa_d=\widehat\kappa/365=0.01137$ and
$\eta_d=\widehat\eta/\sqrt{365}=0.06503$. In normalized form, the strict
Feller inequality is $2\widehat\kappa>\widehat\eta^2$.

For the pricing experiment, the normalized rate is initialized from its
invariant distribution,
\[
    v_0
    \sim
    \operatorname{Gamma}(\alpha_v,\alpha_v),
    \qquad
    \alpha_v
    =
    \frac{2\widehat\kappa}
    {\widehat\eta^2}
    =
    5.3777,
\]
where the second Gamma parameter is the rate. This initialization gives
$\mathbb E[v_t]=1$ and stationary increments of accumulated operational
time. Changing $\widehat\theta$ while holding
$\widehat\sigma_\lambda$ fixed remains economically relevant because it
changes $\widehat\eta$.

\subsection{Persistent-scale and contract calibration}

The cumulative-index strike
$K=-28.5\,^{\circ}\mathrm C$-days is used as the empirical lower-tail
anchor for a 90-day contract. To ensure reproducibility and allow
independent assessment of this calibration target, the sample of 90-day
cumulative anomalies was constructed as follows. Overlapping (rolling)
90-day windows were used to maximize the sample size, advancing the start
date by one day at each step. The coverage window was strictly restricted
to the winter season, defined as December~1 through February~28 (or 29),
and only windows fully contained within these dates were retained. In the
1990--2020 Chicago O'Hare record, leap days (February~29) were excluded
from the daily anomaly series prior to forming the cumulative sums, and
any sporadic missing daily observations were filled by linear interpolation
to maintain a continuous sequence. This procedure yielded a total sample
of $N=2699$ overlapping 90-day cumulative observations. The 5th
percentile was computed using simple order statistics (specifically, the
$\lfloor 0.05 \times N \rfloor + 1$-th smallest value), which produced
the strike $K=-28.5\,^{\circ}\mathrm C$-days.

Under the increment-based model and
$\bar\mu_{90}=0$, the persistent amplitude is determined from
\[
    \mathbb E^v\left[
        \Phi_{\mathcal N}\left(
            \frac{-28.5}
            {\sigma_*\Delta(\tau_{90}/t_0)^{0.78}}
        \right)
    \right]
    =
    0.05.
\]
Using the stationary normalized CIR distribution and the time convention
specified above gives
$\sigma_*=0.508566\,^{\circ}\mathrm C$. The remaining baseline contract
parameters are $L=30$ and $\gamma=0.12$. Because this procedure matches a
single cumulative-tail target, it should not be interpreted as joint
statistical identification of $\sigma_*$, $H$, and the CIR parameters.

\section{Numerical audit}\label{app:audit}

This appendix reports checks of the unit conversion, stationary CIR moments,
payoff sign, limiting cases, Monte Carlo error, time discretization, and
conditional Gaussian pricing identity. Sampling error and discretization
error are treated separately, and no covariance-quadrature error arises under
the increment-based cumulative-index formulation.

\begin{table}[H]
\centering
\small
\caption{Principal numerical audit quantities for the 90-day baseline.}
\label{tab:numerical_audit}
\begin{tabular}{lr}
\toprule
Audit quantity & Value\\
\midrule
Theoretical mean of $\tau_{90}$ & 90.000000\\
Simulated mean of $\tau_{90}$ & 89.901935\\
Theoretical standard deviation of $\tau_{90}$ & 33.181008\\
Simulated standard deviation of $\tau_{90}$ & 33.085101\\
Persistent scale $\sigma_*$ & 0.508566\\
Exercise probability & 0.050000\\
Policy-limit probability & 0.002050\\
Expected payment $\mathbb E[\Phi_L]$ & 0.461220\\
Entropic premium $\Pi_\gamma(\Phi_L)$ & 1.709056\\
Monte Carlo standard error & 0.007758\\
Entropic risk loading & 1.247836\\
\bottomrule
\end{tabular}
\end{table}

For a stationary normalized CIR rate,
$\operatorname{Var}(v_0)=\eta_d^2/(2\kappa_d)$ and
\[
    \operatorname{Var}(\tau_T)
    =
    2\operatorname{Var}(v_0)
    \left[
        \frac{T}{\kappa_d}
        -
        \frac{1-e^{-\kappa_dT}}{\kappa_d^2}
    \right].
\]
At $T=90$, the exact standard deviation is $33.1810$, compared with the
simulated value $33.0851$. The simulated mean $89.9019$ is also consistent
with its theoretical value of 90 within Monte Carlo uncertainty. The
exercise probability equals the imposed tail target, confirming the payoff
sign and the persistent-scale calibration.

\subsection{Inequalities and limiting cases}

The baseline values satisfy the entropic Jensen bound:
\begin{equation}\label{eq:numerical_jensen_check}
    \Pi_\gamma(\Phi_L)
    =
    1.709056
    \geq
    \mathbb E[\Phi_L]
    =
    0.461220.
\end{equation}
The difference $1.247836$ agrees with the risk loading reported in
Table~\ref{tab:numerical_audit}. At $\gamma=10^{-4}$, the premium is
$0.461585$, confirming the limit
$\Pi_\gamma(\Phi_L)\to\mathbb E[\Phi_L]$ as
$\gamma\downarrow0$.

If $K\to-\infty$, the payoff converges pointwise to zero; using
$K=-10{,}000$ gives a premium equal to zero at the reported precision.
If $\sigma_*\to0$, the cumulative index converges to
$\bar\mu_T$ and the premium converges to
$\min\{(K-\bar\mu_T)^+,L\}$, which is zero under the baseline values.

As the normalized time-change volatility tends to zero, its invariant
distribution degenerates at one, so $v_t\to1$ and $\tau_T\to T$. The
TC-fBm premium therefore converges to the deterministic-time pure-fBm
premium. At the baseline $H=0.78$, this limiting premium is $0.878831$.
For $H=1/2$, the conditional index variance reduces to
$\sigma_*^2\Delta^2(\tau_T/t_0)$, consistently with the Brownian
time-change limit.

Finally, $0\leq\Phi_L\leq30$ and $\gamma=0.12$ imply
\begin{equation}\label{eq:numerical_exponential_bound}
    1
    \leq
    e^{\gamma\Phi_L}
    \leq
    e^{3.6}
    \approx
    36.60.
\end{equation}
Thus, the exact exponential kernel is uniformly bounded throughout the
reported parameter grids.

\subsection{Time-discretization and path-count stability}

The CIR discretization is examined with $50{,}000$ stationary
time-change realizations. Increasing the resolution from one to eight steps
per calendar day gives
\[
\begin{array}{c|rrrr}
\text{Steps/day} & 1 & 2 & 4 & 8\\ \hline
M & 90 & 180 & 360 & 720\\
\mathbb E[\tau_{90}]
&89.881&89.943&89.958&89.980\\
\widehat\Pi_{\gamma,N}
&1.7127&1.7156&1.7113&1.7128\\
\widehat{\operatorname{se}}
&0.0110&0.0110&0.0110&0.0110
\end{array}
\]
The differences between resolutions are smaller than the corresponding
Monte Carlo uncertainty and show no systematic discretization drift. Four
steps per day are therefore sufficient for the reported precision.

Using the first $N$ realizations from the same 100,000-path experiment gives
\[
\begin{array}{c|rrr}
N & 10{,}000 & 50{,}000 & 100{,}000\\ \hline
\widehat\Pi_{\gamma,N}
&1.6986&1.7226&1.7091\\
\widehat{\operatorname{se}}
&0.0247&0.0111&0.0078
\end{array}
\]
The standard-error ratios are consistent with the theoretical
$N^{-1/2}$ rate. Because the estimates are nested, their pointwise
differences should not be interpreted as independent replications; the
reported intervals are used to assess whether the observed fluctuations are
compatible with sampling uncertainty.

\subsection{Direct terminal-simulation cross-validation}

The cumulative increment identity implies
\[
    I_T
    =
    \bar\mu_T
    +
    \sigma_*\Delta
    \left(\frac{\tau_T}{t_0}\right)^H Z,
    \qquad
    Z\sim N(0,1),
\]
conditional on $\tau_T$. This provides a direct simulation benchmark for
the analytical conditional kernel. For each of $100{,}000$ independently
simulated time-change realizations, 20 independent standard normal draws are
used to simulate the terminal index and capped payoff. In this procedure, the
20 inner draws per path are treated as a single cluster. Let
 $\bar Y_n = \frac{1}{20}\sum_{i=1}^{20} Y_{n,i}$ denote the cluster mean
for the $n$-th simulated path, and let
 $\bar Y = \frac{1}{N}\sum_{n=1}^N \bar Y_n$ be the overall sample mean.
The variance of the cluster-mean estimator is given by the sample variance
of the cluster means:
\begin{equation}\label{eq:cluster_var}
    \widehat{\mathrm{Var}}(\bar Y) 
    = 
    \frac{1}{N(N-1)}\sum_{n=1}^N\big(\bar Y_n - \bar Y\big)^2.
\end{equation}
Because the reported premium is the log-transformed quantity 
 $\hat\Pi = \frac{1}{\gamma}\log\bar Y$, the delta-method step, rather than
the cluster-mean variance alone, gives the reported premium standard error.
Applying the same delta-method structure used in Proposition~4.3 for the
non-clustered estimator yields
\begin{equation}\label{eq:cluster_delta_var}
    \widehat{\mathrm{Var}}(\hat\Pi) 
    \approx 
    \frac{\widehat{\mathrm{Var}}(\bar Y)}{\gamma^2\bar Y^2}.
\end{equation}
Taking the square root of Eq.~\eqref{eq:cluster_delta_var} provides the
standard error reported for the direct simulation.

\begin{table}[H]
\centering
\caption{Cross-validation of the conditional Gaussian kernel against direct
terminal simulation.}
\label{tab:cross_validation}
\begin{tabular}{lcc}
\toprule
Method & Premium & Monte Carlo s.e.\\
\midrule
Semi-analytical conditional kernel & 1.70906 & 0.00776\\
Direct terminal simulation & 1.70843 & 0.01261\\
\bottomrule
\end{tabular}
\end{table}

The direct and semi-analytical estimates differ by $0.00063$. Accounting
for their paired use of the same time-change realizations gives a standard
error of $0.00995$ for this difference, so the discrepancy is only
 $0.063$ paired standard errors. The two valuation procedures therefore
agree well within simulation uncertainty.

\bibliographystyle{plain}

\end{document}